\pdfoutput=1
\documentclass[11pt]{article}

\usepackage[a4paper,margin=2.5cm]{geometry}
\usepackage[T1]{fontenc}
\usepackage{lmodern}
\usepackage{amsmath,amssymb,amsthm,bm}
\usepackage{graphicx}
\usepackage{booktabs}
\usepackage{array}
\usepackage{tabularx}
\usepackage{seqsplit}
\usepackage{url}
\usepackage[numbers,square,sort&compress]{natbib}
\usepackage[hidelinks]{hyperref}
\usepackage[section]{placeins}
\newcolumntype{L}[1]{>{\raggedright\arraybackslash}p{#1}}
\graphicspath{{./}}

\numberwithin{equation}{section}

\newtheorem{theorem}{Theorem}
\newtheorem{proposition}{Proposition}
\newtheorem{corollary}{Corollary}
\newtheorem{lemma}{Lemma}
\theoremstyle{remark}
\newtheorem{remark}{Remark}

\title{A finite-shot mean-squared-error boundary for Richardson\\ zero-noise extrapolation}

\author{Vicenzo Scavino Alfaro\\[2pt]
\small Independent Researcher, Lima, Peru\\
\small \href{mailto:vicenzoscavino@hotmail.com}{vicenzoscavino@hotmail.com}\quad ORCID: \href{https://orcid.org/0009-0000-2472-9785}{0009-0000-2472-9785}}

\date{}

\begin{document}

\maketitle

\begin{abstract}
Zero-noise extrapolation (ZNE) reduces the bias of noisy quantum-circuit measurements, but Richardson coefficients and shot splitting increase their statistical error. This paper determines the noise strength at which Richardson ZNE starts to lower the mean-squared error relative to the unmitigated estimator at a fixed shot budget $B$. The boundary is set by the balance between the leading squared-bias improvement and the leading excess-variance penalty: depending on the variance exponent of the measurement, it shrinks with $B$, becomes a budget threshold, or disappears. For fixed Richardson extrapolation with linear leading bias, the Richardson order and the shot allocation change only constants; the boundary exponent is set by the measured observable and protocol. Deterministic stabilizer measurements and nonzero-variance Hamiltonian or variational-energy measurements have boundaries that shrink as $B^{-1}$ and $B^{-1/2}$, and for finite-outcome shot measurements these two classes are the only ones. Qiskit Aer simulations of Greenberger--Horne--Zeilinger (GHZ) and quantum approximate optimization algorithm (QAOA) circuits reproduce this separation, and on IBM Quantum hardware ZNE lowers the mean-squared error at equal budget in 75 of 101 instances.
\end{abstract}

\noindent\textbf{Keywords:} zero-noise extrapolation, quantum error mitigation, finite-shot sampling, Richardson extrapolation, quantum computing, QAOA

\section{Introduction}
\label{sec:intro}

Richardson extrapolation originated as a classical extrapolation method for finite-difference and numerical integration problems \cite{Richardson1911,BulirschStoer1966}. In quantum error mitigation (QEM), zero-noise extrapolation (ZNE) evaluates a circuit at amplified noise levels and extrapolates the measured expectation value to the zero-noise limit. Richardson ZNE forms a linear combination of finite-shot estimators whose coefficients cancel low-order terms in the noise expansion \cite{Temme2017,LiBenjamin2017,Endo2018,Cai2023,Endo2021}; practical noise-scaling methods include unitary folding, local/global folding, layerwise variants, and identity insertion \cite{GiurgicaTiron2020,He2020,Pascuzzi2022,RussoMari2024}.

Finite-shot budgets make bias reduction insufficient as a decision criterion. Richardson coefficients may have magnitude larger than one, include negative weights, and split the shot budget across several noise levels, so ZNE can reduce systematic bias while increasing sampling variance. This finite-resource issue is related to QEM benchmarking, direct Richardson analyses, sampling-cost lower bounds, and error-scaling analyses \cite{Krebsbach2022,MohammadipourLi2025,Bultrini2023,Takagi2022,Takagi2023,Tsubouchi2023,Qin2023,Quek2024}. The operational question studied here is: under a fixed shot budget, at what noise strength does fixed Richardson ZNE start to reduce the mean-squared error (MSE) relative to the unmitigated estimator?

We call the answer the finite-shot MSE boundary: the first positive small-noise zero of
\begin{equation}
\Delta_{\mathrm{MSE}}(\epsilon,B)
=\mathrm{MSE}_{\mathrm{noisy}}(\epsilon,B)
-\mathrm{MSE}_{\mathrm{ZNE}}(\epsilon,B),
\label{eq:delta-mse}
\end{equation}
where $\epsilon$ is the physical noise strength and $B$ the total shot budget; ZNE lowers the MSE where $\Delta_{\mathrm{MSE}}>0$. The boundary is set by the competition between the squared bias that ZNE removes and the sampling variance that it adds, a trade-off familiar from finite-shot mitigation, probabilistic error cancellation, and sampling-overhead analyses \cite{Piveteau2021,Piveteau2022,Suzuki2022}. The main result makes this competition quantitative: if the leading squared-bias improvement is $D_p\epsilon^{2p}$ and the leading excess-variance penalty is $K_q\epsilon^q/B$, then
\begin{equation}
\epsilon^*(B)\sim (K_q/D_p)^{1/(2p-q)}B^{-1/(2p-q)}\qquad(0\le q<2p),
\end{equation}
while $q=2p$ gives a budget threshold and $q>2p$ gives no shrinking lower boundary. For fixed Richardson ZNE with linear leading bias, $p=1$ at every Richardson order, so the order and the shot allocation enter only through $K_q$; the exponent is set by the variance exponent $q$ of the measured observable and protocol, which separates deterministic stabilizer measurements ($q=1$, boundary $\propto B^{-1}$) from nonzero-variance energy measurements ($q=0$, boundary $\propto B^{-1/2}$).

\noindent\textbf{Contributions.}
\begin{itemize}
\item[(i)] A local theorem for the MSE crossing of any estimator pair with the leading balance above, covering the shrinking-boundary, budget-threshold, and no-boundary regimes; in the shrinking regime the crossing is the only sign change in a fixed small-noise interval, with an explicit convergence rate and a finite-budget bracket (Section~\ref{sec:master-theorem}, Appendices~\ref{app:master-proof} and~\ref{app:aux-boundary}).
\item[(ii)] A derivation of this balance for fixed Richardson ZNE from primitive bias, variance, coefficient, and allocation conditions, which identifies $D_p$ and $K_q$ (Section~\ref{sec:richardson}).
\item[(iii)] A classification of the variance exponent for finite-outcome shot measurements, $q=0$ for a nonzero ideal variance and $1\le q\le p$ for a deterministic ideal outcome under smooth noise, which makes these two classes exhaustive for linear leading bias, excludes the critical and supercritical regimes under independent sampling, and gives the boundary constant of deterministic binary observables, such as Greenberger--Horne--Zeilinger (GHZ) stabilizers, from the bias slope alone (Section~\ref{sec:consequences}).
\item[(iv)] A procedure to extract the boundary exponent and constant from data, applied to Qiskit Aer simulations of GHZ and quantum approximate optimization algorithm (QAOA) circuits, with synthetic controls of the remaining regimes and robustness and failure cases, and to IBM Quantum hardware runs analyzed at equal total budget (Sections~\ref{sec:diagnostic} and~\ref{sec:empirical}).
\end{itemize}

Figure~\ref{fig:boundary-flow} summarizes the logic of the paper: local bias and variance expansions, the Richardson rule, and the shot split produce the two competing leading terms whose balance defines the boundary.

\begin{figure}[t]
\centering
\includegraphics[width=0.96\textwidth,trim=0 30 0 0,clip]{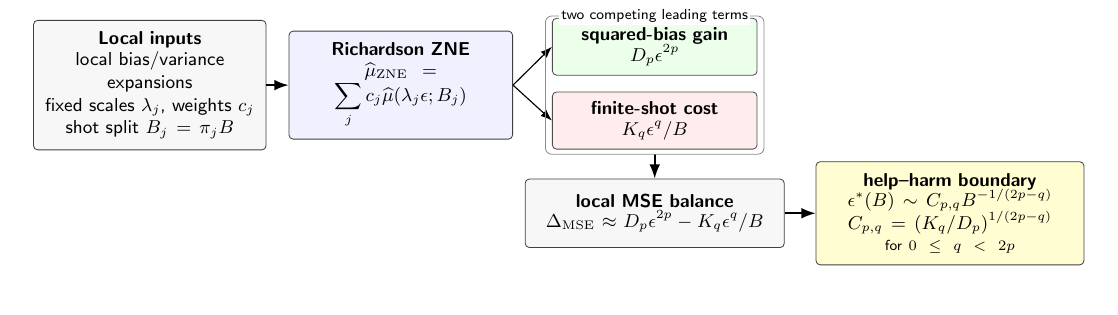}
\caption{Operational schematic of the finite-shot ZNE boundary. Primitive local bias and variance expansions, fixed Richardson scales and coefficients, and a fixed shot split define the ZNE estimator. The leading squared-bias gain and finite-shot excess-variance cost produce the local MSE balance, whose lower crossing gives the help--harm boundary when $0\le q<2p$.}
\label{fig:boundary-flow}
\end{figure}

\section{Setup and assumptions}
\label{sec:setup}

Let $O$ be a bounded observable with ideal expectation $\mu_0=\langle O\rangle_0$ and noisy expectation $\mu(\epsilon)=\langle O\rangle_\epsilon$. A fixed measurement protocol $\mathcal P$ turns each shot into a random variable $Y_\epsilon^{(\mathcal P)}$ with single-shot variance $v(\epsilon)=\mathrm{Var}(Y_\epsilon^{(\mathcal P)})$. For a single Pauli observable, $Y_\epsilon^{(\mathcal P)}\in\{-1,+1\}$ \cite{NielsenChuang2010}; for Hamiltonians measured through several Pauli terms, the effective variance also depends on grouping, allocation, covariance, and normalization.

The unmitigated estimator $\widehat\mu_{\mathrm{noisy}}$ uses all $B$ shots at noise $\epsilon$, and the ZNE estimator $\widehat\mu_{\mathrm{ZNE}}$ combines estimates at amplified noise levels (Section~\ref{sec:richardson}). For either estimator, $\mathrm{MSE}(\widehat\mu)=\mathbb E[(\widehat\mu-\mu_0)^2]=\mathrm{Bias}(\widehat\mu)^2+\mathrm{Var}(\widehat\mu)$. Define
\begin{equation}
\begin{aligned}
D(\epsilon)&:=\mathrm{Bias}_{\mathrm{noisy}}(\epsilon)^2-
\mathrm{Bias}_{\mathrm{ZNE}}(\epsilon)^2,\\
A(\epsilon)&:=B\left[\mathrm{Var}(\widehat\mu_{\mathrm{ZNE}})-
\mathrm{Var}(\widehat\mu_{\mathrm{noisy}})\right].
\end{aligned}
\label{eq:D-A}
\end{equation}
Then the MSE difference~\eqref{eq:delta-mse} is exactly
\begin{equation}
\Delta_{\mathrm{MSE}}(\epsilon,B)=D(\epsilon)-\frac{A(\epsilon)}{B}.
\label{eq:exact-decomposition}
\end{equation}
Here $D$ is the squared bias removed by ZNE and $A/B$ the variance it adds; for sample means with fixed allocation and independent sampling, $A$ does not depend on $B$ (Section~\ref{sec:variance-penalty}). When $v(0)=0$, $\epsilon=0$ is a trivial zero of $\Delta_{\mathrm{MSE}}$; the boundary is the first positive crossing.

\textbf{A1: leading bias.} For some order $p\ge1$, uniformly over the finite scale set $\{\lambda_j\}_{j=0}^{k}$ of A2,
\begin{equation}
\mu(\lambda_j\epsilon)=\mu_0+a_p\lambda_j^p\epsilon^p+O(\epsilon^{p+\delta_b}),
\qquad a_p\ne0,\quad \delta_b>0.
\label{eq:A1}
\end{equation}
The main case is linear leading bias, $p=1$ and $a_1=\alpha\ne0$.

\textbf{A2: fixed linear rule.} The scale factors $1=\lambda_0<\lambda_1<\cdots<\lambda_k$, the coefficients $c_j$ with $\sum_jc_j=1$, and the allocation $n_j=\pi_jB$, with $\pi_j>0$ and $\sum_j\pi_j=1$, are fixed.

\textbf{A3: independent sampling.} Samples used at distinct scale factors are independent; correlated protocols change the variance constant (Remark~\ref{rem:rounding-correlated}).

\textbf{A4: variance exponent.} For some $q\ge0$, the single-shot variance $v_j(\epsilon)$ at the $j$th implemented noise scale satisfies, uniformly over the scale set,
\begin{equation}
v_j(\epsilon)=\nu_j\epsilon^q+O(\epsilon^{q+\delta_v}),
\qquad \nu_j>0,\quad \delta_v>0.
\end{equation}
The exponent $q$ is a property of the observable and measurement protocol, not of the extrapolation rule; ideal multiplicative noise scaling is the special case $\nu_j=\nu\lambda_j^q$. On $\mu_0$ the subscript $0$ denotes zero noise; on $c_0$, $\pi_0$, $v_0$, and $\nu_0$ it denotes the unamplified scale $\lambda_0=1$, so $v_0(\epsilon)=v(\epsilon)$, while $v(0)$ is the ideal variance.

\textbf{A5: leading balance.} For some $p\ge1$, $q\ge0$, $D_p>0$, $K_q>0$, and $\delta_b,\delta_v>0$,
\begin{equation}
D(\epsilon)=D_p\epsilon^{2p}+O(\epsilon^{2p+\delta_b}),
\qquad
A(\epsilon)=K_q\epsilon^q+O(\epsilon^{q+\delta_v}),
\label{eq:A5-expansions}
\end{equation}
so that, by Eq.~\eqref{eq:exact-decomposition},
\begin{equation}
\begin{aligned}
\Delta_{\mathrm{MSE}}(\epsilon,B)
&=D_p\epsilon^{2p}-\frac{K_q\epsilon^q}{B}+R_p(\epsilon,B),\\
|R_p(\epsilon,B)|&\le L_b\epsilon^{2p+\delta_b}+L_v\frac{\epsilon^{q+\delta_v}}{B}
\qquad(0<\epsilon\le\epsilon_{\mathrm R}),
\end{aligned}
\label{eq:leading-mse-balance}
\end{equation}
with constants $\epsilon_{\mathrm R},L_b,L_v$ that do not depend on $B$. In the subcritical case $0\le q<2p$ we set $\omega:=(2p-q)^{-1}$ and define the rescaled remainder $\eta_B(x):=B^{2p\omega}R_p(xB^{-\omega},B)$, which by Eq.~\eqref{eq:leading-mse-balance} tends to zero uniformly for $x$ in compact subsets of $(0,\infty)$. We also require $\partial_x\eta_B\to0$ there, which holds when the two remainders in Eq.~\eqref{eq:A5-expansions} are differentiable with derivatives $O(\epsilon^{2p+\delta_b-1})$ and $O(\epsilon^{q+\delta_v-1})$.

A1--A4 are primitive conditions on the Richardson construction. A5 is the only hypothesis of the general theorem in Section~\ref{sec:master-theorem}, and Corollary~\ref{cor:richardson-main} derives it from A1--A4 for fixed Richardson ZNE. Table~\ref{tab:assumptions} lists the role of each assumption and the consequences of its failure.

\begin{table}[t]
\centering
\small
\caption{Assumptions, their role, and consequences of failure.}
\label{tab:assumptions}
\begin{tabularx}{\textwidth}{@{}L{0.36\textwidth}L{0.30\textwidth}X@{}}
\toprule
Assumption & Role in the theorem & If it fails \\
\midrule
$D_p>0$; for fixed Richardson with $p=1$ this follows from $\alpha\ne0$ & Gives the leading squared-bias benefit $D_p\epsilon^{2p}$ & Boundary exponent changes, reverses, or disappears \\
A4: $v_j(\epsilon)=\nu_j\epsilon^q+O(\epsilon^{q+\delta_v})$ uniformly over the scale set & Determines the leading variance penalty & Boundary may change scale or fail to exist \\
$K_q>0$ (automatic in the multiplicative model; an assumption for scale-specific $\nu_j$) & Ensures a positive leading ZNE variance cost & Variance-penalty interpretation changes \\
A3: independent samples across scale factors & Removes covariance terms & Constants or leading terms can change \\
A1: leading bias expansion and controlled noise scaling & Fixes the bias order $p$ and, with the rule, $D_p$ & Extrapolation error can dominate \\
A5: leading MSE-balance expansion with remainder control & Gives uniform root stability and explicit rate & Rate or even asymptotic root control may fail \\
\bottomrule
\end{tabularx}
\end{table}

\section{Local MSE-boundary theorem}
\label{sec:master-theorem}

The lower boundary is controlled by the first nonzero power in the squared-bias improvement and the first nonzero power in the excess-variance penalty, not by the extrapolation order itself.

\begin{theorem}[Local MSE-boundary theorem]
\label{thm:master}
Assume A5.

\textup{(i) Subcritical regime, $0\le q<2p$.} Let
\begin{equation}
C_{p,q}=\left(\frac{K_q}{D_p}\right)^{1/(2p-q)}.
\label{eq:master-constant}
\end{equation}
There is $\bar\epsilon>0$ such that, for all sufficiently large $B$, $\Delta_{\mathrm{MSE}}(\cdot,B)$ has exactly one zero $\epsilon^*(B)$ in $(0,\bar\epsilon]$. It is sign-changing: $\Delta_{\mathrm{MSE}}<0$ (ZNE has the larger MSE) on $(0,\epsilon^*(B))$ and $\Delta_{\mathrm{MSE}}>0$ (ZNE has the smaller MSE) on $(\epsilon^*(B),\bar\epsilon]$. Moreover,
\begin{equation}
\epsilon^*(B)
=C_{p,q}B^{-1/(2p-q)}\left[1+O(B^{-\theta})\right],
\qquad
\theta=\frac{\min\{\delta_b,\delta_v\}}{2p-q}.
\label{eq:master-law}
\end{equation}

\textup{(ii) Critical regime, $q=2p$.} For each fixed $B$,
\begin{equation}
\Delta_{\mathrm{MSE}}(\epsilon,B)
=\left(D_p-\frac{K_q}{B}\right)\epsilon^{2p}+o(\epsilon^{2p}),
\end{equation}
so the local sign is governed by the budget threshold $B^*=K_q/D_p$: for $B>B^*$, ZNE lowers the MSE for all sufficiently small $\epsilon>0$; for $B<B^*$, it raises the MSE for all sufficiently small $\epsilon>0$; the case $B=B^*$ is decided by higher-order terms.

\textup{(iii) Supercritical regime, $q>2p$.} For each fixed $B>0$, $\Delta_{\mathrm{MSE}}(\epsilon,B)>0$ for all sufficiently small $\epsilon>0$: the variance penalty vanishes too fast to create a lower boundary that shrinks to zero.
\end{theorem}

In the subcritical regime, $\epsilon^*(B)$ is therefore the first positive crossing of Eq.~\eqref{eq:delta-mse}, and above it ZNE lowers the MSE throughout a fixed small-noise interval. The proof (Appendix~\ref{app:master-proof}) works on the scale $\epsilon=xB^{-\omega}$, where the rescaled difference $B^{2p\omega}\Delta_{\mathrm{MSE}}(xB^{-\omega},B)$ converges to the balanced polynomial $f_{p,q}(x)=D_px^{2p}-K_qx^q$; its simple root $C_{p,q}$ gives the crossing, while below and above this scale the variance penalty and the bias improvement, respectively, dominate. Appendix~\ref{app:aux-boundary} brackets $\epsilon^*(B)$ at finite $B$ with an explicit budget threshold. Table~\ref{tab:regimes} summarizes the three regimes and Figure~\ref{fig:boundary-mechanism} illustrates them.

\begin{table}[t]
\centering
\small
\caption{Local regimes for the finite-shot MSE boundary.}
\label{tab:regimes}
\begin{tabularx}{\textwidth}{@{}L{0.16\textwidth}L{0.42\textwidth}X@{}}
\toprule
Regime & Local result & Interpretation \\
\midrule
$0\le q<2p$ & $\epsilon^*(B)\sim C_{p,q}B^{-1/(2p-q)}$, Eq.~\eqref{eq:master-constant} & Standard shrinking lower boundary \\
$q=2p$ & $B^*=K_q/D_p$ & Budget threshold, not a shrinking boundary law \\
$q>2p$ & No leading-order shrinking lower boundary & Variance penalty is too weak near zero \\
\bottomrule
\end{tabularx}
\end{table}

\begin{figure}[t]
\centering
\includegraphics[width=0.96\textwidth]{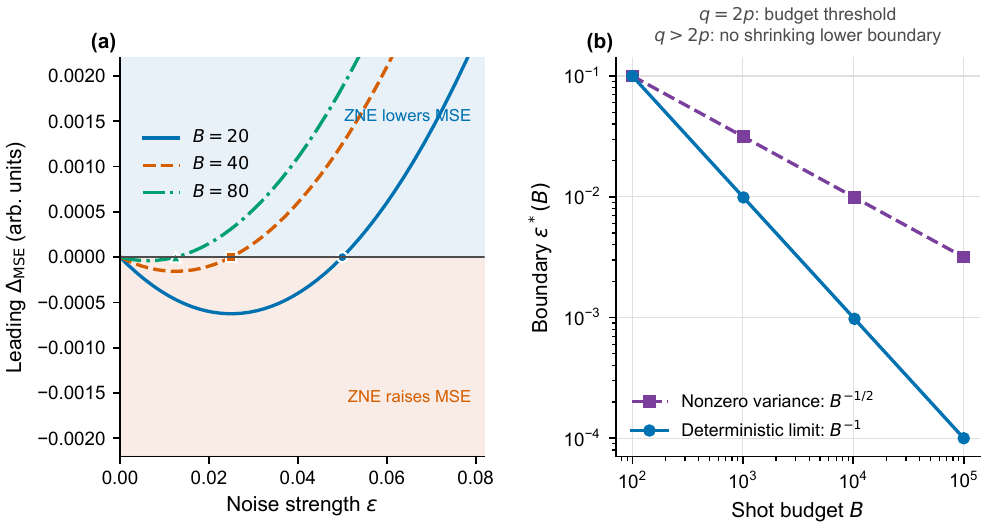}
\caption{Leading-order mechanism of the finite-shot MSE boundary. (a) Schematic subcritical MSE differences for $p=q=1$ and three shot budgets; the plotted units are arbitrary and illustrate how the first positive crossing separates locally larger-MSE and lower-MSE ZNE regimes. (b) Representative $p=1$ boundary classes: nonzero-variance measurements with $q=0$ give $B^{-1/2}$, whereas deterministic-limit measurements with $q=1$ give $B^{-1}$.}
\label{fig:boundary-mechanism}
\end{figure}

\section{Fixed Richardson ZNE}
\label{sec:richardson}

This section derives A5 for fixed Richardson ZNE from A1--A4, identifies $D_p$ and $K_q$, and applies Theorem~\ref{thm:master}.

\subsection{Richardson rule and bias improvement}
\label{sec:richardson-bias}

A Richardson rule of order $k\ge1$ chooses the coefficients of A2 so that
\begin{equation}
\sum_{j=0}^k c_j\lambda_j^m=0,
\qquad
m=1,\ldots,k;
\label{eq:richardson-identities}
\end{equation}
together with $\sum_jc_j=1$, these identities cancel the noise expansion of the estimator mean through order $k$ when $\mu$ is smooth \cite{Richardson1911,Temme2017}. The estimator is
\begin{equation}
\widehat\mu_{\mathrm{ZNE}}
=\sum_{j=0}^k c_j\widehat\mu(\lambda_j\epsilon),
\end{equation}
where $\widehat\mu(\lambda_j\epsilon)$ is the sample mean of $n_j=\pi_jB$ shots at noise $\lambda_j\epsilon$. Because $\sum_jc_j=1$, A1 gives the two biases
\begin{equation}
\begin{aligned}
\mathrm{Bias}_{\mathrm{noisy}}(\epsilon)&=a_p\epsilon^p+O(\epsilon^{p+\delta_b}),\\
\mathrm{Bias}_{\mathrm{ZNE}}(\epsilon)&=\sum_jc_j\bigl[\mu(\lambda_j\epsilon)-\mu_0\bigr]=\rho_pa_p\epsilon^p+O(\epsilon^{p+\delta_b}),
\end{aligned}
\end{equation}
where $\rho_p:=\sum_jc_j\lambda_j^p$, so that
\begin{equation}
D(\epsilon)=D_p\epsilon^{2p}+O(\epsilon^{2p+\delta_b}),
\qquad
D_p=a_p^2(1-\rho_p^2).
\label{eq:Dp}
\end{equation}
For integer $p\le k$, Eq.~\eqref{eq:richardson-identities} gives $\rho_p=0$ and hence $D_p=a_p^2$. In particular, linear leading bias gives $p=1$ and $D_1=\alpha^2$ for every Richardson order $k\ge1$.

\subsection{Variance penalty}
\label{sec:variance-penalty}

Under A2 and A3, $\mathrm{Var}(\widehat\mu_{\mathrm{ZNE}})=B^{-1}\sum_jc_j^2v_j(\epsilon)/\pi_j$ and $\mathrm{Var}(\widehat\mu_{\mathrm{noisy}})=v(\epsilon)/B$, so the excess-variance function of Eq.~\eqref{eq:D-A} is
\begin{equation}
A(\epsilon)=\sum_{j=0}^k \frac{c_j^2}{\pi_j}v_j(\epsilon)-v(\epsilon),
\label{eq:Ak}
\end{equation}
independent of $B$. Under A4,
\begin{equation}
A(\epsilon)=K_q\epsilon^q+O(\epsilon^{q+\delta_v}),
\qquad
K_q=\sum_{j=0}^k\frac{c_j^2\nu_j}{\pi_j}-\nu_0,
\label{eq:Kgen}
\end{equation}
and in the multiplicative case $\nu_j=\nu\lambda_j^q$ this is the closed-form Richardson constant
\begin{equation}
K_{q,k}
=\nu\left[
\sum_{j=0}^k\frac{c_j^2\lambda_j^q}{\pi_j}-1
\right].
\label{eq:Kqk}
\end{equation}

\begin{proposition}[Richardson variance penalty]
\label{prop:positive-variance-penalty}
Under A2 and Eq.~\eqref{eq:richardson-identities}, $K_{q,k}>0$ for every $\nu>0$ and $q\ge0$.
\end{proposition}

\begin{proof}
A rule with $\sum_jc_j\lambda_j=0$ and $\lambda_j\ge1$ cannot have all $c_j\ge0$, since then $\sum_jc_j\lambda_j\ge\sum_jc_j=1$; hence $\sum_j|c_j|>\sum_jc_j=1$. Using $\lambda_j^q\ge1$ and the Cauchy--Schwarz inequality for $|c_j|=\sqrt{\pi_j}\,|c_j|/\sqrt{\pi_j}$ with $\sum_j\pi_j=1$,
\begin{equation}
\sum_j\frac{c_j^2\lambda_j^q}{\pi_j}
\ge\sum_j\frac{c_j^2}{\pi_j}
\ge\Bigl(\sum_j|c_j|\Bigr)^2>1.
\end{equation}
\end{proof}

With scale-specific constants $\nu_j$, as produced by implemented folding, or with covariance between scale estimates, positivity of $K_q$ is not automatic. For first-order Richardson with scale factors $(1,a)$, $a>1$, the coefficients are $c_0=a/(a-1)$ and $c_1=-1/(a-1)$; for the common choice $(1,3)$ with uniform allocation, $K_{q,1}=\nu[7/2+3^q/2]$.

\subsection{Boundary law}
\label{sec:main-theorem}

\begin{corollary}[Fixed Richardson boundary]
\label{cor:richardson-main}
Assume A1--A4 with $|\rho_p|<1$ and $K_q>0$ (automatic for a Richardson rule in the multiplicative case, by Proposition~\ref{prop:positive-variance-penalty}), and that the bias and variance remainders are differentiable with the derivative bounds of A5. Then A5 holds with $D_p=a_p^2(1-\rho_p^2)$ and $K_q$ from Eq.~\eqref{eq:Kgen}, and Theorem~\ref{thm:master} applies. In particular, for a Richardson rule of order $k\ge1$ and linear leading bias $a_1=\alpha\ne0$: if $0\le q<2$,
\begin{equation}
\epsilon^*(B)\sim C_{1,q}B^{-1/(2-q)},
\qquad
C_{1,q}=\left(\frac{K_q}{\alpha^2}\right)^{1/(2-q)};
\label{eq:richardson-law}
\end{equation}
if $q=2$, the budget threshold is $B^*=K_q/\alpha^2$; and if $q>2$, there is no shrinking lower boundary.
\end{corollary}

\begin{proof}
Equations~\eqref{eq:Dp} and~\eqref{eq:Kgen} give the two expansions of A5, with $D_p>0$ because $|\rho_p|<1$, and the remainder conditions are assumed. Theorem~\ref{thm:master} then gives the three cases; for $p=1$, Eq.~\eqref{eq:richardson-identities} gives $\rho_1=0$ and hence $D_1=\alpha^2$.
\end{proof}

If $|\rho_p|\ge1$, the rule does not improve the leading squared bias, and the lower boundary may disappear, reverse, or move outside the perturbative window. The Richardson rule and the allocation enter $C_{1,q}$ only through $K_q$.

\begin{remark}[Richardson order and exponent]
\label{rem:order-exponent}
For the boundary defined against the unmitigated estimator, increasing the Richardson order does not change the exponent when $\alpha\ne0$: the improvement stays $\alpha^2\epsilon^2$ because the unmitigated estimator keeps its first-order bias, while the mitigated bias vanishes to higher order ($O(\epsilon^{k+1})$ for smooth $\mu$), and the order enters only through the coefficients and scale factors in $K_{q,k}$. The exponent $-1/(2(k+1)-q)$ belongs to a different comparison, for example one balancing the residual ZNE bias against the ZNE variance.
\end{remark}

\subsection{Shot allocation, rounding, and correlated sampling}
\label{sec:optimal-allocation}

\begin{remark}[Optimal shot allocation]
\label{rem:optimal-allocation}
Under independent sampling and positive scale variances at a fixed $\epsilon>0$, the Neyman allocation $n_j^{\mathrm{opt}}=B|c_j|\sqrt{v_j(\epsilon)}/\sum_\ell|c_\ell|\sqrt{v_\ell(\epsilon)}$ minimizes the ZNE variance at fixed total budget \cite{Neyman1934}, with minimum $B^{-1}\bigl(\sum_j|c_j|\sqrt{v_j(\epsilon)}\bigr)^2$. Under A4 the corresponding penalty is $K_q^{\mathrm{opt}}\epsilon^q+O(\epsilon^{q+\delta_v})$ with
\begin{equation}
K_q^{\mathrm{opt}}
=\Bigl(\sum_j |c_j|\sqrt{\nu_j}\Bigr)^2-\nu_0
=\nu\Bigl[\Bigl(\sum_j |c_j|\lambda_j^{q/2}\Bigr)^2-1\Bigr]
\quad\text{if }\nu_j=\nu\lambda_j^q.
\end{equation}
The multiplicative form is positive because $\lambda_j^{q/2}\ge1$ and $\sum_j|c_j|>1$ (proof of Proposition~\ref{prop:positive-variance-penalty}); for scale-specific constants $\nu_j$, positivity of $K_q^{\mathrm{opt}}$ is again not automatic. The optimizer depends on $\epsilon$ and is therefore not a fixed allocation in the sense of A2, but Theorem~\ref{thm:master} uses only the expansion of $A(\epsilon)$ in A5, and the limiting weights $\pi_j^{\mathrm{opt},0}=|c_j|\sqrt{\nu_j}/\sum_\ell |c_\ell|\sqrt{\nu_\ell}$ give the same leading constant. Optimal allocation thus changes the boundary constant but not the exponent.
\end{remark}

\begin{remark}[Rounding and correlated sampling]
\label{rem:rounding-correlated}
If $n_j=\pi_jB+O(1)$ with fixed weights $\pi_j>0$, then $1/n_j=(\pi_jB)^{-1}+O(B^{-2})$, so rounding adds $O(\epsilon^q/B)$ to $A$ when $v_j(\epsilon)=O(\epsilon^q)$, that is, $O(\epsilon^q/B^2)$ to $R_p$, with the corresponding derivative bound. Relative to the leading terms this is $O(B^{-1})$ in each region of the proof of Theorem~\ref{thm:master}, so its conclusions are unchanged, with $\theta$ replaced by $\min\{\theta,1\}$. The same bound on the rounding error holds for $\epsilon$-dependent weights that converge to positive limits, such as the optimal weights above; the resulting excess variance must then still satisfy A5.

Correlated sampling is handled by replacing the independent variance penalty with the effective covariance penalty. If
\begin{equation}
\mathrm{Var}(\widehat\mu_{\mathrm{ZNE}})
=B^{-1}c^\top\Gamma_{\mathrm{ZNE}}(\epsilon)c+o(B^{-1}),
\qquad
\mathrm{Var}(\widehat\mu_{\mathrm{noisy}})
=B^{-1}\Gamma_{\mathrm{noisy}}(\epsilon)+o(B^{-1}),
\end{equation}
the excess variance is $A(\epsilon)=c^\top\Gamma_{\mathrm{ZNE}}(\epsilon)c-\Gamma_{\mathrm{noisy}}(\epsilon)$, and Theorem~\ref{thm:master} applies unchanged whenever this effective penalty satisfies A5 with $K_q>0$.
\end{remark}

\section{Observable classes}
\label{sec:consequences}

Corollary~\ref{cor:richardson-main} reduces the boundary exponent to the variance exponent $q$ of the measured protocol. For shot-based measurements, $q$ is not a free parameter: it is fixed by whether the ideal measured variable is deterministic.

\begin{proposition}[Variance-exponent classes]
\label{prop:variance-exponents}
Fix the circuit size and the measurement protocol, let the single-shot variable $Y_\epsilon$ take values in a finite set, and assume A1 and A4.

\textup{(i)} If the ideal variance is positive, $v(0)>0$, then $q=0$ and $\nu_j=v(0)$. A4 itself then holds whenever $v_j(\epsilon)=v(0)+O(\epsilon^{\delta_v})$ for some $\delta_v>0$, for instance if the $v_j$ are differentiable at zero.

\textup{(ii)} If $v(0)=0$, that is, if the ideal state gives a deterministic outcome $y_0=\mu_0$, then $q\le p$; if, in addition, $v$ is Lipschitz at zero, then $1\le q\le p$.

\textup{(iii)} If, in addition to (ii), $Y_\epsilon\in\{-1,+1\}$, then $q=p$ and $\nu_j=2|a_p|\lambda_j^p$.
\end{proposition}

\begin{proof}
For (i), every scale is noiseless at $\epsilon=0$, so each $v_j(\epsilon)$ tends to $v(0)>0$ and A4 forces $q=0$ and $\nu_j=v(0)$. For (ii), let $d>0$ be the smallest distance between $y_0$ and the other possible outcomes, and $b(\epsilon)=\mu(\epsilon)-\mu_0$. Since $|Y_\epsilon-y_0|\ge d$ whenever $Y_\epsilon\ne y_0$,
\begin{equation}
v(\epsilon)=\mathbb E\bigl[(Y_\epsilon-y_0)^2\bigr]-b(\epsilon)^2
\ge d\,\mathbb E|Y_\epsilon-y_0|-b(\epsilon)^2
\ge d\,|b(\epsilon)|-b(\epsilon)^2.
\label{eq:gap-bound}
\end{equation}
By A1, $|b(\epsilon)|=|a_p|\epsilon^p(1+o(1))$, so $v(\epsilon)\ge\tfrac12d|a_p|\epsilon^p$ for small $\epsilon$, which together with A4 forces $q\le p$; a Lipschitz $v$ with $v(0)=0$ is $O(\epsilon)$, which forces $q\ge1$. For (iii), $|\mu|\le1$ forces $\mu_0=\sigma\in\{-1,+1\}$ and $\sigma a_p<0$, so $v_j(\epsilon)=1-\mu(\lambda_j\epsilon)^2=2|a_p|\lambda_j^p\epsilon^p+O(\epsilon^{p+\delta'})$ with $\delta'=\min\{\delta_b,p\}$.
\end{proof}

For linear leading bias, $p=1$, and smooth noise ($v$ Lipschitz at zero), parts (i) and (ii) leave exactly two classes: $q=0$ when the ideal measured variable has nonzero variance, as for nonzero-variance energy measurements in QAOA/MaxCut \cite{Farhi2014,Farhi2022} and in variational algorithms \cite{Peruzzo2014,McClean2016,Kandala2017,Preskill2018,Cerezo2021,Bharti2022}, and $q=1$ when it is deterministic, as for GHZ stabilizers \cite{Greenberger1989,Mermin1990,AaronsonGottesman2004}. For direct measurement of a cost Hamiltonian $H_C$, $v(0)=\langle H_C^2\rangle_0-\langle H_C\rangle_0^2$, which is positive unless the ideal state is an eigenstate of $H_C$. By Corollary~\ref{cor:richardson-main}, the boundary of fixed Richardson ZNE therefore shrinks either as $C_{1,0}B^{-1/2}$ or as $C_{1,1}B^{-1}$, and no other exponent occurs for finite-outcome shot measurements with smooth noise and independent sampling. More generally, part (ii) gives $q\le p<2p$, so with independent sampling and $K_q>0$ such measurements realize only the subcritical regime of Theorem~\ref{thm:master}; the critical and supercritical regimes require an excess variance that vanishes faster than the single-shot variance, as can happen with correlated sampling (Remark~\ref{rem:rounding-correlated}), or effective estimators that are not averages of finite-outcome shots. For grouped multi-Pauli protocols, whose effective variance is a positive combination of group variances, the bound~\eqref{eq:gap-bound} holds group by group and the same conclusion follows. Readout error illustrates the classification: it makes the ideal outcome of a stabilizer measurement random, so $v(0)>0$; calibration with the known readout probability removes the bias at $\epsilon=0$ but keeps this variance floor, which moves the measurement from $q=1$ to $q=0$ (Section~\ref{sec:robustness}).

\begin{proposition}[Deterministic binary observables: the GHZ class]
\label{prop:depol-contraction-a5}
Let $Y_\epsilon\in\{-1,+1\}$ have a deterministic ideal value, and assume A1 with $p=1$ and slope $\alpha\ne0$, with a $C^1$ remainder whose derivative is $O(\epsilon)$, together with A2 and A3. Then, for every Richardson rule of order $k\ge1$, A4 holds with $q=1$ and $\nu_j=2|\alpha|\lambda_j$, A5 holds with $p=q=1$, $D_1=\alpha^2$, and $K_1=K_{1,k}$ from Eq.~\eqref{eq:Kqk}, and
\begin{equation}
\epsilon^*(B)\sim\frac{2}{|\alpha|}\left[\sum_j\frac{c_j^2\lambda_j}{\pi_j}-1\right]B^{-1}.
\label{eq:ghz-boundary}
\end{equation}
For the scale factors $(1,3)$ with uniform allocation, $\epsilon^*(B)\sim10/(|\alpha|B)$.
\end{proposition}

\begin{proof}
Proposition~\ref{prop:variance-exponents}(iii) with $p=1$ gives A4 with $q=1$ in the multiplicative form $\nu_j=\nu\lambda_j$ with $\nu=2|\alpha|$, so the penalty constant is $K_{1,k}=2|\alpha|[\sum_jc_j^2\lambda_j/\pi_j-1]$, positive by Proposition~\ref{prop:positive-variance-penalty}. Through $v=1-\mu^2$, the variance remainder inherits the $C^1$ bound of the bias remainder. Corollary~\ref{cor:richardson-main} with $q=1$ gives $\epsilon^*(B)\sim(K_{1,k}/\alpha^2)B^{-1}$, which is Eq.~\eqref{eq:ghz-boundary}; for $(1,3)$ with uniform allocation the bracket equals $5$.
\end{proof}

For an $n$-qubit GHZ state measured in $X^{\otimes n}$, local Pauli contraction noise whose adjoint acts on each measured factor as $\mathcal N_\epsilon^\dagger(X)=(1-g\epsilon+O(\epsilon^2))X+R_\epsilon$, with the $R_\epsilon$ terms contributing only $O_n(\epsilon^2)$ to $\langle X^{\otimes n}\rangle$, gives $\mu(\epsilon)=(1-g\epsilon)^n+O_n(\epsilon^2)$ and hence $|\alpha|=ng$; more generally, $\ell_n$ Pauli-active noisy locations with contraction $g$ give $|\alpha|=g\ell_n$. Under ideal noise scaling, Eq.~\eqref{eq:ghz-boundary} then holds with this slope. All constants are pointwise in circuit size.

\section{Boundary extraction from data}
\label{sec:diagnostic}

For $p=1$ and $0\le q<2$, Corollary~\ref{cor:richardson-main} predicts the asymptotic log-log slope
\begin{equation}
s=\lim_{B\to\infty}\frac{\log\epsilon^*(B)}{\log B}=-\frac1{2-q}.
\label{eq:predicted-slope}
\end{equation}
The most direct test estimates $q$ independently of the boundary and compares $\widehat s_{\mathrm{pred}}=-1/(2-\widehat q)$ with the observed boundary slope $\widehat s_{\mathrm{obs}}$.

\textbf{Local bias coefficient and variance exponent.} On a fixed small-noise grid we fit
\begin{equation}
\mu(\epsilon_i)-\mu_0=\alpha\epsilon_i+\beta\epsilon_i^2+u_i,
\qquad
\log v(\epsilon_i)=\log\nu+q\log\epsilon_i+w_i.
\label{eq:local-fits}
\end{equation}
For higher-order cases, the bias regression uses the lowest nonzero power allowed by the construction. The variance-regression window is fixed independently of the boundary fit; when $v(0)>0$ is available from the statevector, it is used directly to assign $q=0$ rather than relying only on a log-log slope.

\textbf{Constants.} Inserting $\widehat\alpha$, $\widehat\nu$, and $\widehat q$ into Eq.~\eqref{eq:Kqk} and Eq.~\eqref{eq:richardson-law} gives the plug-in estimates $\widehat K_{q,k}$ and $\widehat C$ of $K_{q,k}$ and $C_{1,q}$; for scale-specific constants, $\widehat K_q=\sum_j c_j^2\widehat\nu_j/\pi_j-\widehat\nu_0$ replaces $\widehat K_{q,k}$. The constants reported in Section~\ref{sec:main-evidence} use the multiplicative form; that section also reports the effect of the scale-specific one. The fitted intercept of $\log\epsilon^*(B)=\log C+s\log B+e_B$ is compared with the plug-in $\widehat C$.

\textbf{Boundary estimation and uncertainty.} For each budget $B$, the crossing is the interpolated zero at the first transition of $\Delta_{\mathrm{MSE}}$ from $\le0$ to $>0$ inside a fixed small-noise window, rather than the first positive grid point; budgets without such a transition are reported as having no crossing in the window. Bootstrap resampling acts on the raw counts, preserving circuit, noise-scale, seed, and observable labels; each replicate recomputes the noisy and ZNE estimates, $\Delta_{\mathrm{MSE}}$, the crossings, $\widehat q$, $\widehat s_{\mathrm{obs}}$, and $\widehat C^{\mathrm{fit}}$ over the fixed list of budgets. The supplement contains these per-seed bootstrap intervals (1000 replicates). The slope intervals shown in Section~\ref{sec:main-evidence} are instead 95\% $t$-intervals across the three seed-level fits, which also include seed-to-seed variation. For the central depolarizing fits the two are comparable: the across-seed standard deviation of the slope is $0.3$--$0.8$ times the mean per-seed bootstrap standard deviation, except for GHZ $n=3$ ($2.6$, driven by one seed).

\textbf{Plug-in computation of $\Delta_{\mathrm{MSE}}$.} In the Aer study, $\Delta_{\mathrm{MSE}}(\epsilon,B)$ is computed from calibrated moments rather than from repeated $B$-shot experiments. For each noise level $\epsilon$ and scale factor $\lambda_j$, a calibration run with $N_{\mathrm{cal}}$ shots ($N_{\mathrm{cal}}=2\times10^5$, or $8\times10^5$ in the GHZ $n=7$ follow-up of Section~\ref{sec:sim-setup}) gives plug-in estimates $\bar y_j(\epsilon)$ and $\widehat v_j(\epsilon)$ of the single-shot mean and variance, and the ideal value $\mu_0$ is computed exactly from the noiseless statevector. The finite-budget errors are then evaluated analytically,
\begin{equation}
\begin{aligned}
\widehat{\mathrm{MSE}}_{\mathrm{noisy}}(\epsilon,B)&=\bigl(\bar y_0(\epsilon)-\mu_0\bigr)^2+\frac{\widehat v_0(\epsilon)}{B},\\
\widehat{\mathrm{MSE}}_{\mathrm{ZNE}}(\epsilon,B)&=\Bigl(\sum_j c_j\bar y_j(\epsilon)-\mu_0\Bigr)^2+\sum_j\frac{c_j^2\,\widehat v_j(\epsilon)}{n_j},
\end{aligned}
\label{eq:plugin-mse}
\end{equation}
so that the budget enters only through the variance terms, and $\epsilon^*(B)$ is the first sign change of their difference. Bootstrap uncertainty reflects raw-count resampling of the calibration propagated through Eq.~\eqref{eq:plugin-mse}; the crossings at different budgets share the same calibrated moments and are therefore correlated. Section~\ref{sec:robustness} compares this shortcut with direct Monte Carlo repetitions.

Because $\mathbb E[(\bar y_j-\mu(\lambda_j\epsilon))^2]=v_j/N_{\mathrm{cal}}$, the plug-in squared biases are biased upward, and the plug-in difference carries the budget-independent offset $(v_0-\sum_jc_j^2v_j)/N_{\mathrm{cal}}$, which is negative for the $[1,3]$ rule used here. With $B\le16000\ll N_{\mathrm{cal}}$ this offset is small; it slightly enlarges $\epsilon^*$ at the largest budgets and flattens the fitted slopes. Removing it shifts the central depolarizing slopes by $0.003$--$0.012$ toward the predicted values (GHZ $n=3,5$: from $-0.9705,-0.9801$ to $-0.9768,-0.9871$; QAOA $n=6,8$: from $-0.4961,-0.4955$ to $-0.5004,-0.4985$). The slopes reported below are the uncorrected, slightly conservative values.

\section{Numerical evidence from Qiskit Aer and IBM hardware}
\label{sec:empirical}

The simulations test whether the boundary exponent separates by variance class, using Qiskit \cite{JavadiAbhari2024}, Mitiq-compatible folding \cite{LaRose2022}, and raw-count bootstrap resampling \cite{Efron1979,DavisonHinkley1997}. Qiskit Aer provides the exponent evidence, with $\epsilon$, the budgets, the interpolation, and the resampling under control; the IBM Quantum runs test the implementation and the equal-budget MSE comparison on hardware. Unless stated otherwise, Aer results use first-order Richardson ZNE with scale factors $[1,3]$, gate folding, and uniform shot allocation. Depolarizing GHZ and QAOA are the central cases; amplitude damping, readout and state-preparation-and-measurement (SPAM) errors, allocation, and depth variations probe robustness and the limits of the assumptions (Table~\ref{tab:assumption-diagnostics}).

\subsection{Simulation setup}
\label{sec:sim-setup}

Circuits are first transpiled to the basis $\{R_x,R_z,\mathrm{CX}\}$ and then folded with Mitiq's random gate folding \cite{LaRose2022}; at $\lambda=3$ every gate is folded exactly once, so the folded circuit is deterministic and has three times the gates of the original. Depolarizing noise acts after every $R_x$ gate with probability $\epsilon$ and after every CX gate with probability $2\epsilon$ ($R_z$ gates are noiseless); amplitude damping uses $\gamma=\epsilon$ on each qubit of the same gates, and coherent over-rotation applies an extra $R_x(2\epsilon)$ to each qubit of the same gates. The noise strength takes 20 logarithmically spaced values $\epsilon\in[10^{-5},10^{-1}]$, and the budgets are $B\in\{500,1000,2000,4000,8000,16000\}$. GHZ circuits with $n=3,5,7$ qubits are measured in the $X$ basis and estimate $\langle X^{\otimes n}\rangle$, whose ideal value is $1$. QAOA circuits are depth-one MaxCut circuits on rings of $n=6$ and $n=8$ vertices with fixed, non-optimized parameters $\gamma=0.8$, $\beta=0.6$; the observable is the cut value $c$ normalized as $(c-m/2)/(m/2)$ for $m$ edges, whose ideal mean is $-0.338$ and whose ideal variance is nonzero. Ideal values come from the noiseless statevector. Calibration uses $N_{\mathrm{cal}}=2\times10^5$ shots per noise level and scale factor and three simulator seeds. Two further depolarizing runs differ only as stated: a fine-grid run for GHZ $n=5$ and QAOA $n=6$ with 25 noise levels and the additional budget $B=32000$, and a follow-up for GHZ $n=7$ with $N_{\mathrm{cal}}=8\times10^5$ and the additional budget $B=32000$ (Section~\ref{sec:main-evidence}). The software versions are Qiskit 2.3.1, Qiskit Aer 0.17.2, and Mitiq 1.0.0.

\begin{table}[t]
\centering
\footnotesize
\setlength{\tabcolsep}{3pt}
\renewcommand{\arraystretch}{1.08}
\caption{Assumption diagnostics for the numerical tests.}
\label{tab:assumption-diagnostics}
\begin{tabularx}{\textwidth}{@{}L{0.16\textwidth}L{0.21\textwidth}L{0.24\textwidth}L{0.15\textwidth}X@{}}
\toprule
Case & Expected class & Small-noise diagnostic & Crossing & Use \\
\midrule
GHZ, depol. & deterministic binary & $\widehat q\approx0.99$ locally & negative-to-positive & central; near $B^{-1}$ \\
QAOA, depol. & nonzero ideal variance & $v(0)>0$; $\widehat q\approx0$ locally & negative-to-positive & central; near $B^{-1/2}$ \\
GHZ, amp. damping & deterministic binary & exact-moment test near $B^{-1}$ & negative-to-positive & robustness \\
QAOA, amp. damping & nonzero ideal variance & shifted finite-window slope & negative-to-positive & robustness \\
Readout/SPAM & changed measurement protocol & known-readout-probability calibration shifts $\widehat q$ & negative-to-positive after calibration & regime change/control \\
\bottomrule
\end{tabularx}
\end{table}

\subsection{Central depolarizing evidence}
\label{sec:main-evidence}

\textbf{Boundary slopes.} Figure~\ref{fig:boundary-extraction} shows the boundary extraction for representative GHZ and QAOA cases: at fixed budget, the first interpolated sign change of $\Delta_{\mathrm{MSE}}$ defines the crossing, and across budgets the crossings follow distinct GHZ and QAOA power laws. Figure~\ref{fig:key-exponents} and Table~\ref{tab:main-aer} collect the fitted slopes. Under depolarizing noise, the GHZ slopes lie within $0.03$ of $-1$ and the QAOA slopes within $0.005$ of $-1/2$, except GHZ $n=7$ in the main run ($-0.9418$); repeating this case at higher calibration (the follow-up of Section~\ref{sec:sim-setup}) gives $-0.9863$. The amplitude-damping entries are discussed in Section~\ref{sec:robustness}.

\begin{figure}[t]
\centering
\includegraphics[width=0.96\textwidth]{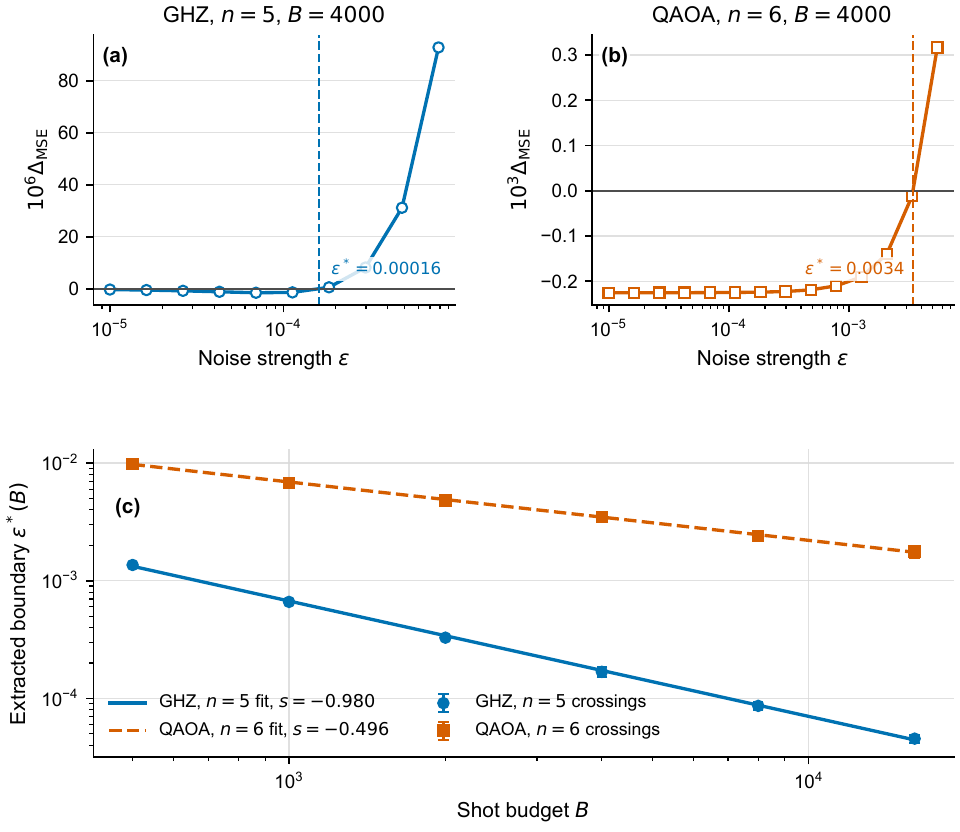}
\caption{Representative extraction of the MSE boundary from the central depolarizing Aer data. Panels (a) and (b) show $\Delta_{\mathrm{MSE}}$ for GHZ $n=5$ and QAOA $n=6$ at $B=4000$, uniform allocation, and fixed simulation seed; dashed vertical lines mark the interpolated first negative-to-positive crossing. Panel (c) shows the six-budget boundary curves as means over three seeds, with one-standard-deviation error bars. The fitted slopes are $-0.980$ for GHZ and $-0.496$ for QAOA.}
\label{fig:boundary-extraction}
\end{figure}

\begin{figure}[t]
\centering
\includegraphics[width=0.92\textwidth]{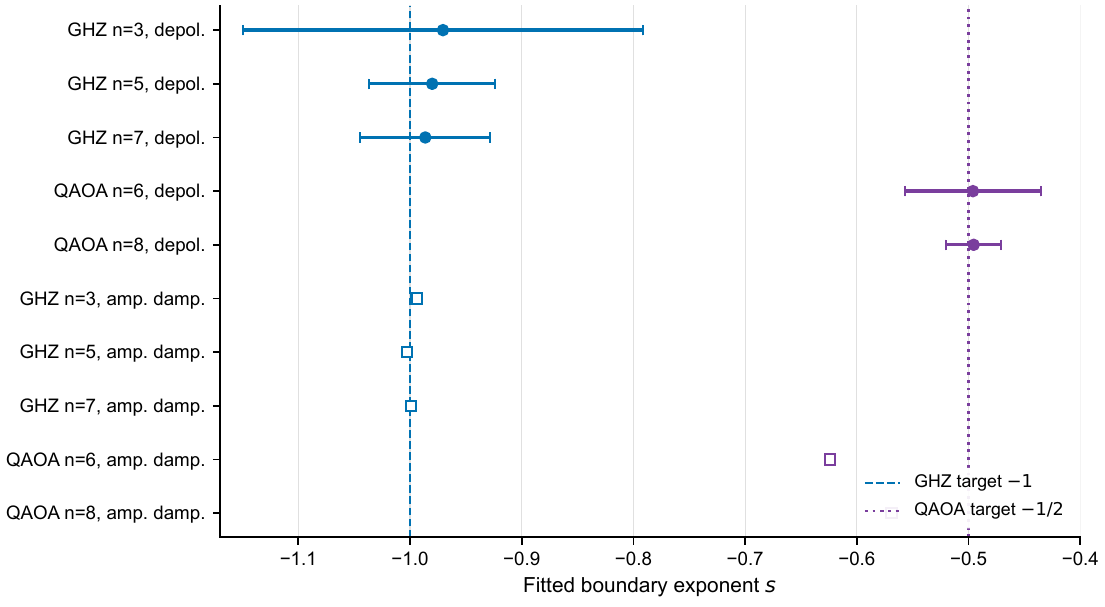}
\caption{Key fitted MSE-boundary exponents from Qiskit Aer simulations. Filled circles show the central depolarizing cases (for GHZ $n=7$, the follow-up with $N_{\mathrm{cal}}=8\times10^5$) with 95\% $t$-intervals computed across three seed-level slope fits. Open squares show the amplitude-damping cases, computed from exact moments; the three folding seeds give identical slopes, so no interval is shown.}
\label{fig:key-exponents}
\end{figure}

\begin{table}[t]
\centering
\small
\setlength{\tabcolsep}{4pt}
\caption{Fitted slopes of the lower MSE boundary in Qiskit Aer. Depolarizing entries are means over three uniform-allocation seeds for the runs of Section~\ref{sec:sim-setup}; amplitude-damping entries are exact-moment computations, identical across the three folding seeds.}
\label{tab:main-aer}
\begin{tabularx}{\textwidth}{@{}L{0.24\textwidth}L{0.16\textwidth}Xr@{}}
\toprule
Noise model & Case & Run & Slope $s$ \\
\midrule
depolarizing & GHZ $n=5$ & fine grid & $-0.9966$ \\
depolarizing & QAOA $n=6$ & fine grid & $-0.5010$ \\
depolarizing & GHZ $n=3$ & main & $-0.9705$ \\
depolarizing & GHZ $n=5$ & main & $-0.9801$ \\
depolarizing & GHZ $n=7$ & main & $-0.9418$ \\
depolarizing & GHZ $n=7$ & follow-up & $-0.9863$ \\
depolarizing & QAOA $n=6$ & main & $-0.4961$ \\
depolarizing & QAOA $n=8$ & main & $-0.4955$ \\
amplitude damping & GHZ $n=3$ & exact moments & $-0.9940$ \\
amplitude damping & GHZ $n=5$ & exact moments & $-1.0028$ \\
amplitude damping & GHZ $n=7$ & exact moments & $-0.9990$ \\
amplitude damping & QAOA $n=6$ & exact moments & $-0.6242$ \\
amplitude damping & QAOA $n=8$ & exact moments & $-0.5689$ \\
\bottomrule
\end{tabularx}
\end{table}

\textbf{Variance exponent versus boundary slope.} Figure~\ref{fig:q-to-s} and Table~\ref{tab:q-to-s} compare the observed slopes with $\widehat s_{\mathrm{pred}}=-1/(2-\widehat q)$, where $\widehat q$ comes from the variance curve (Section~\ref{sec:diagnostic}). The two differ by less than $0.02$ in every case except GHZ $n=7$ in the main run, where the difference is $0.049$; in the follow-up it is $0.0035$.

Because $\widehat q$, the crossings, and the plug-in constants all come from the same calibrated moments (Section~\ref{sec:diagnostic}), this comparison tests whether the simulated circuits satisfy the hypotheses of the theorem---a linear leading bias and an excess variance that inherits the exponent of the single-shot variance---and whether the scanned budgets are already asymptotic. In the follow-up, the crossing at the largest budget, $B=32000$, lies one or two grid points above the lower edge of the $\epsilon$ window; omitting that budget gives $-0.9929$, even closer to the prediction.

\textbf{Precision of individual slopes.} Each slope is fitted to six budgets spanning a factor of $32$, and the intervals in Figures~\ref{fig:key-exponents} and~\ref{fig:q-to-s} rest on three seeds. All 15 seed-level fits of the main run have $R^2>0.98$ and log-log residuals without a monotone trend; their 95\% bootstrap intervals for the slope are $0.081$--$0.166$ wide, and three seed-level slopes (one GHZ $n=3$ seed and two GHZ $n=7$ seeds) lie more than $0.05$ from the predicted exponent. The GHZ/QAOA separation of about $0.48$ is several times these widths; intervals narrower than $0.1$ would take $1.4$--$2.7$ times more calibration shots per case.

\begin{figure}[t]
\centering
\includegraphics[width=0.94\textwidth]{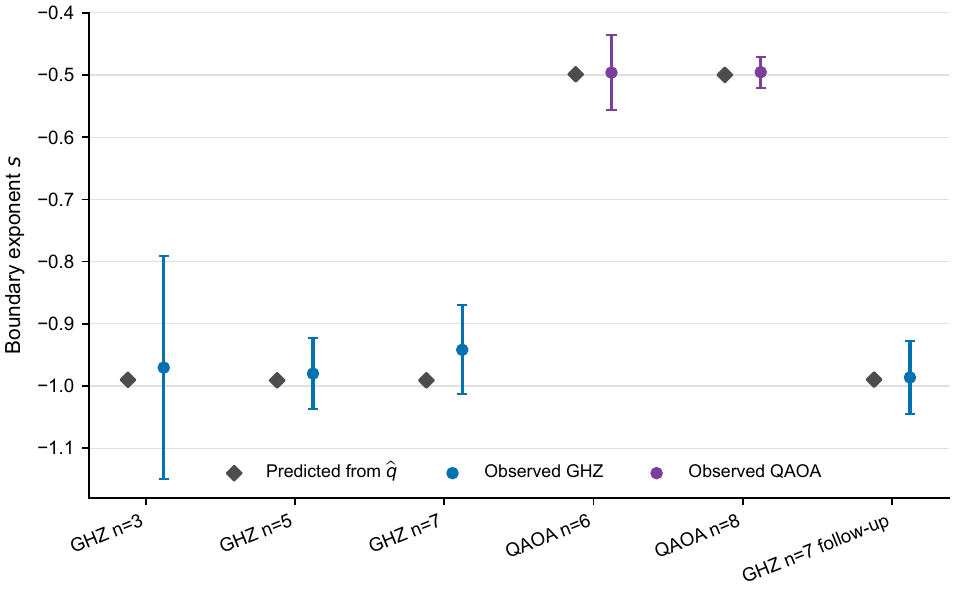}
\caption{Independent variance-exponent prediction of the boundary slope. Diamonds show $\widehat s_{\mathrm{pred}}=-1/(2-\widehat q)$ obtained from variance diagnostics that do not use the boundary slope; circles show the observed boundary slopes with 95\% $t$-intervals computed across three seed-level slope fits. For GHZ $n=7$ both the main run and the follow-up with $N_{\mathrm{cal}}=8\times10^5$ are shown; the deviation of the former is absent in the latter.}
\label{fig:q-to-s}
\end{figure}

\begin{table}[t]
\centering
\footnotesize
\setlength{\tabcolsep}{4pt}
\renewcommand{\arraystretch}{1.08}
\caption{Independent variance-exponent prediction of the observed boundary slope. The small negative QAOA values of $\widehat q$ are finite-window regression artifacts; QAOA is classified as $q=0$ by its nonzero ideal variance.}
\label{tab:q-to-s}
\begin{tabularx}{\textwidth}{@{}L{0.18\textwidth}Xr r r r@{}}
\toprule
Case & Run ($N_{\mathrm{cal}}$) & $\widehat q$ & $s_{\mathrm{pred}}$ & $s_{\mathrm{obs}}$ & $|s_{\mathrm{obs}}-s_{\mathrm{pred}}|$ \\
\midrule
GHZ $n=3$ & main ($2\times10^5$) & 0.989980 & $-0.990080$ & $-0.970479$ & 0.0196 \\
GHZ $n=5$ & main ($2\times10^5$) & 0.990831 & $-0.990914$ & $-0.980117$ & 0.0108 \\
GHZ $n=7$ & main ($2\times10^5$) & 0.990819 & $-0.990902$ & $-0.941823$ & 0.0491 \\
GHZ $n=7$ & follow-up ($8\times10^5$) & 0.989666 & $-0.989771$ & $-0.986290$ & 0.0035 \\
QAOA $n=6$ & main ($2\times10^5$) & $-0.004810$ & $-0.498800$ & $-0.496102$ & 0.0027 \\
QAOA $n=8$ & main ($2\times10^5$) & $-0.000706$ & $-0.499824$ & $-0.495458$ & 0.0044 \\
\bottomrule
\end{tabularx}
\end{table}

\textbf{Constants.} Table~\ref{tab:constants} and Figure~\ref{fig:constant-diagnostic} compare the fitted prefactors with the plug-in predictions. They agree to within $4.7$--$10.4\%$ for GHZ $n=3,5$ and QAOA $n=6,8$; GHZ $n=7$ deviates by $29\%$ in the main run and by $1.9\%$ in the follow-up. Fitting a separate variance constant at $\lambda=3$ shows that the multiplicative assumption is close: the fitted constant lies $2$--$5\%$ below the multiplicative prediction for GHZ and within $1\%$ for QAOA, and the scale-specific form changes $C_{\mathrm{theory}}$ by at most $1.4\%$ while slightly reducing every relative error.

\textbf{Allocation.} All slopes above use uniform allocation. Optimal allocation, extracted from the same outputs, gives slopes $-0.9891$ for GHZ $n=5$ and $-0.4974$ for QAOA $n=6$, consistent with the prediction that allocation changes constants but not exponents.

\textbf{The other regimes.} Exact-moment synthetic models test the regimes that the circuits do not realize and, by Proposition~\ref{prop:variance-exponents}, cannot realize with independent finite-outcome shots. With quadratic leading bias ($p=2$), cancelled by second-order Richardson extrapolation at scale factors $[1,3,5]$, the slopes are $-0.2499$ and $-0.3333$ for $q=0$ and $q=1$, against the predicted $-1/4$ and $-1/3$. At $q=2p$ the budget threshold $B^*=20000$ is recovered. A supercritical control ($p=1$ with $q=3$ and $q=4$, and $q=3$ with an adverse quadratic bias term) gives, for every budget from $10^2$ to $10^5$, a positive $\Delta_{\mathrm{MSE}}$ at the smallest $\epsilon$ and no negative-to-positive crossing in the window, as Theorem~\ref{thm:master}(iii) requires; the only sign change is an upper crossing $\epsilon_{\mathrm{up}}(B)=(\alpha^2B/K_q)^{1/(q-2)}$ that moves away from zero, with fitted log-log slopes $1.0001$ and $0.5001$ against the predicted $1$ and $1/2$.

\begin{table}[t]
\centering
\footnotesize
\setlength{\tabcolsep}{4pt}
\renewcommand{\arraystretch}{1.08}
\caption{Constant-level diagnostic for the finite-budget boundary: plug-in prefactor $C_{\mathrm{theory}}$ from the independently estimated local bias and variance coefficients versus the fitted intercept $C_{\mathrm{fit}}$.}
\label{tab:constants}
\begin{tabularx}{0.62\textwidth}{@{}Xr r r@{}}
\toprule
Case & $C_{\mathrm{theory}}$ & $C_{\mathrm{fit}}$ & Rel.\ error \\
\midrule
GHZ $n=3$ & 1.137060 & 1.047758 & 0.078538 \\
GHZ $n=5$ & 0.654424 & 0.586093 & 0.104414 \\
GHZ $n=7$ & 0.444329 & 0.315430 & 0.290099 \\
GHZ $n=7$ follow-up & 0.442496 & 0.434304 & 0.018512 \\
QAOA $n=6$ & 0.222978 & 0.212534 & 0.046840 \\
QAOA $n=8$ & 0.180531 & 0.170369 & 0.056291 \\
\bottomrule
\end{tabularx}
\end{table}

\begin{figure}[t]
\centering
\includegraphics[width=0.62\textwidth]{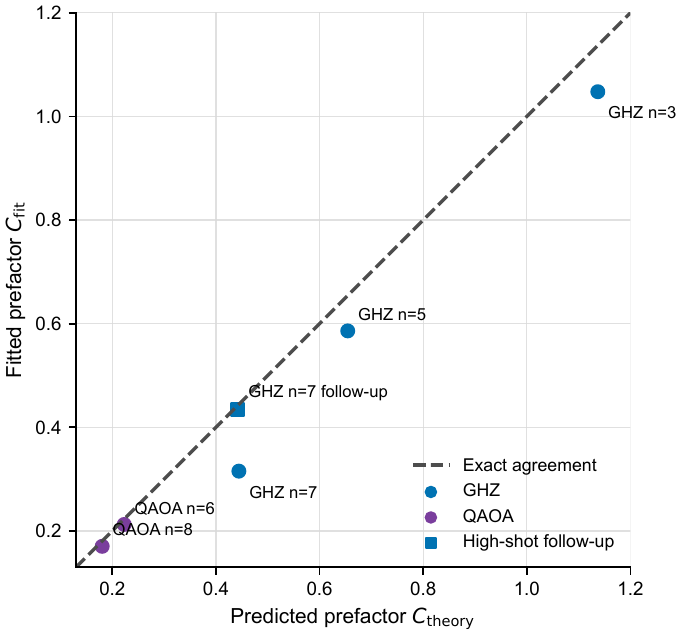}
\caption{Constant-level comparison for the finite-budget boundary: fitted prefactors against the plug-in predictions of Table~\ref{tab:constants}. The dashed line denotes exact agreement; the square marks the GHZ $n=7$ follow-up.}
\label{fig:constant-diagnostic}
\end{figure}

\FloatBarrier
\subsection{Robustness and controls}
\label{sec:robustness}

\textbf{Amplitude damping.} GHZ slopes stay within $0.006$ of $-1$ (Table~\ref{tab:main-aer}). QAOA remains separated from GHZ but is steeper than $-1/2$: fitted over all six budgets, the slopes are $-0.6242$ ($n=6$) and $-0.5689$ ($n=8$). These runs use exact moments, so the deviation is not sampling noise. Local slopes between consecutive budgets show that it is concentrated at the smallest budgets: for $n=6$ they move from $-0.90$ ($B=500\to1000$) to $-0.52$ ($B=8000\to16000$), and for $n=8$ from $-0.71$ to $-0.48$ and $-0.56$ for the last two budget pairs. This is the pattern expected when subleading terms are not yet negligible at the larger crossings reached for small $B$. For example, if the excess variance has the form $A(\epsilon)=K_0+K_1\epsilon+O(\epsilon^2)$ with $K_1>0$, the local slope of the $q=0$ balance is $\bigl[-2+K_1\epsilon^*/(K_0+K_1\epsilon^*)\bigr]^{-1}$, which lies between $-1$ and $-1/2$ and tends to $-1/2$ as $\epsilon^*\to0$. The six-budget QAOA values are therefore finite-window slopes rather than asymptotic exponents, and they preserve the GHZ/QAOA separation.

\textbf{Other variations.} Robustness tests varied the amplification factor $a=2,3,5$, used a quadratic Richardson rule with scale factors $[1,3,5]$, and tested additional GHZ observables. GHZ remained near the $B^{-1}$ class and QAOA near the $B^{-1/2}$ class in every case, and a $Z_0$ GHZ observable acted as a control with no stable finite crossing. Across 20 deterministic random QAOA graphs at $n=8$ (Figure~\ref{fig:qaoa-random}), all 20 slopes lie between $-0.559$ and $-0.492$, with mean $-0.5236$.

\begin{figure}[t]
\centering
\includegraphics[width=0.72\textwidth]{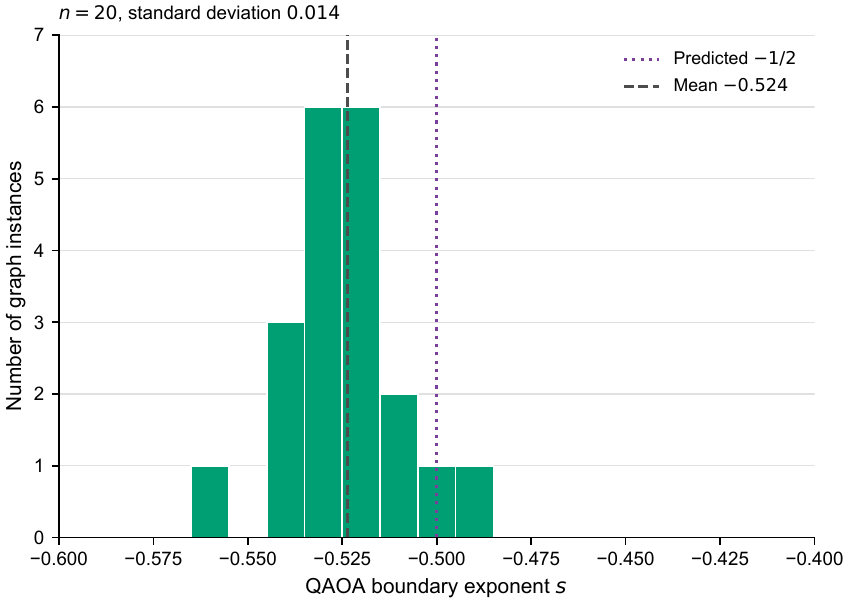}
\caption{Distribution of QAOA boundary slopes for 20 deterministic random graphs with $n=8$. The dotted line marks the predicted $-1/2$ and the dashed line the ensemble mean ($-0.5236$); the panel title reports the population standard deviation (the sample standard deviation is $0.0143$).}
\label{fig:qaoa-random}
\end{figure}

\textbf{QAOA depth.} Depth robustness depends on the parameter schedule. With the fixed QAOA schedule, the $n=8$ boundary slope is $-0.5003$ for $r=1$ and $-0.4957$ for $r=2$, but $-0.2901$ (with $R^2=0.938$) for $r=3$. These fits use the budgets $B\in\{1000,\ldots,64000\}$, evaluated on the same raw counts: on the grid $\{500,\ldots,16000\}$ used elsewhere, the $r=2$ slope is $-0.61$ ($R^2=0.960$), driven by the $B=500$ point, and $r=3$ yields only three crossings, which is why the depth study uses the shifted grid. A scan of four parameter schedules at $r=3$ gave slopes $-0.23$ (the fixed schedule), $-0.68$, $-0.60$, and $-0.41$; after the scan, the depth-adjusted midpoint schedule was rerun over the full budget grid and gave $s=-0.6012$ with $R^2=0.9990$ (Figure~\ref{fig:qaoa-depth}). The failure of the fixed schedule has a direct explanation within the theory. Its ideal value is close to zero ($\mu_0=-0.031$), and its local linear bias coefficient, estimated with Eq.~\eqref{eq:local-fits}, is $|\widehat\alpha|<1$, against $7$--$41$ for the other three schedules. The boundary constant $C_{1,0}=(K_0/\alpha^2)^{1/2}$ is therefore large, and the crossings lie at $\epsilon\approx0.017$--$0.032$, about ten times further out than for the other schedules and outside the perturbative regime in which the leading balance holds. At $r=3$ the slope is therefore schedule-dependent, from $-0.41$ to $-0.68$ among the three schedules with a resolved linear bias.

\begin{figure}[t]
\centering
\includegraphics[width=0.72\textwidth]{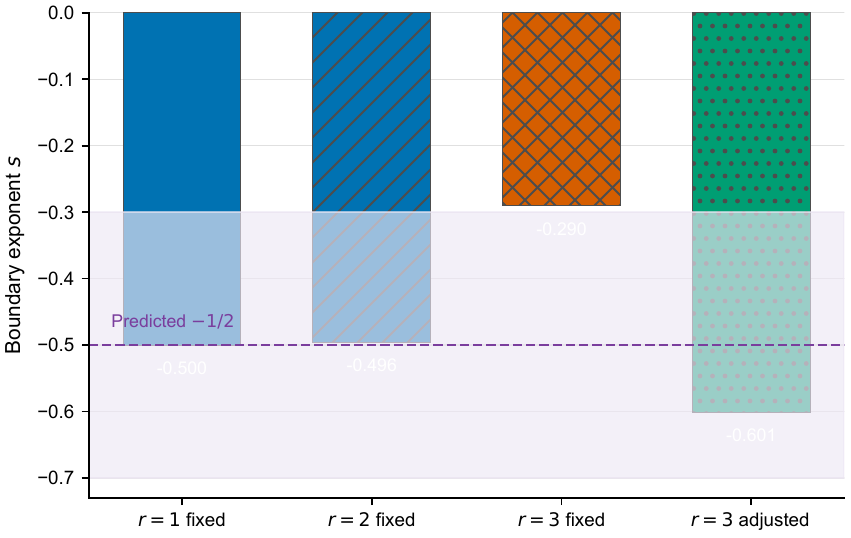}
\caption{QAOA depth study ($n=8$). Bars show the boundary slope for the fixed parameter schedule at depths $r=1,2,3$ and for the depth-adjusted midpoint schedule at $r=3$; the dashed line marks the predicted $-1/2$, and the shaded band, a visual guide, marks slopes within $0.2$ of it. The fixed schedule does not recover the target regime at $r=3$, while the depth-adjusted schedule does.}
\label{fig:qaoa-depth}
\end{figure}

\textbf{Readout/SPAM.} Uncorrected readout error leaves a bias at $\epsilon=0$, which violates the local expansion A1, so the boundary law need not hold. Without correction, at $p_{\mathrm{ro}}=10^{-4},10^{-3},10^{-2}$, GHZ $n=5$ gives slopes $-1.2185$, $-1.1430$, and $-1.0266$, and QAOA $n=8$ gives $-0.4871$, $-0.5391$, and $-0.8287$, the last with a stable boundary; the two GHZ slopes at the smaller readout probabilities and the QAOA slope at the largest deviate clearly from their classes. A simple analytic correction with known readout probability removes the bias at $\epsilon=0$. It brings the QAOA $n=8$ slopes to $-0.4830$, $-0.4809$, and $-0.4914$ at the same three readout probabilities, but GHZ $n=5$ gives $-0.8637$, $-0.6800$, and $-0.4851$: the corrected GHZ slope moves away from $-1$ as the readout probability grows.

Re-estimating $\widehat q$ for the corrected measurement protocol explains these GHZ values. With known-readout-probability calibration, GHZ moves from $q\approx1$ at zero readout probability toward $q\approx0$ for nonzero readout probabilities, as Proposition~\ref{prop:variance-exponents} requires, and the observed slopes follow $\widehat s_{\mathrm{pred}}=-1/(2-\widehat q)$ toward the $B^{-1/2}$ class (Figure~\ref{fig:readout}). For example, at $p_{\mathrm{ro}}=0.01$, GHZ $n=5$ gives $\widehat q=0.080243$, $s_{\mathrm{pred}}=-0.520899$, and $s_{\mathrm{obs}}=-0.550784$, while QAOA $n=6$ gives a near-zero finite-window estimate $\widehat q=-0.004447$, $s_{\mathrm{pred}}=-0.498891$, and $s_{\mathrm{obs}}=-0.513947$. At the transition point (GHZ $n=3$, $p_{\mathrm{ro}}=10^{-3}$) the predicted and observed slopes differ by $0.09$, and the GHZ $n=3$ constant comparison at $p_{\mathrm{ro}}=0.01$ has a relative error of $0.56$. This regime analysis was carried out after the GHZ slopes above were observed.

\begin{figure}[t]
\centering
\includegraphics[width=0.94\textwidth]{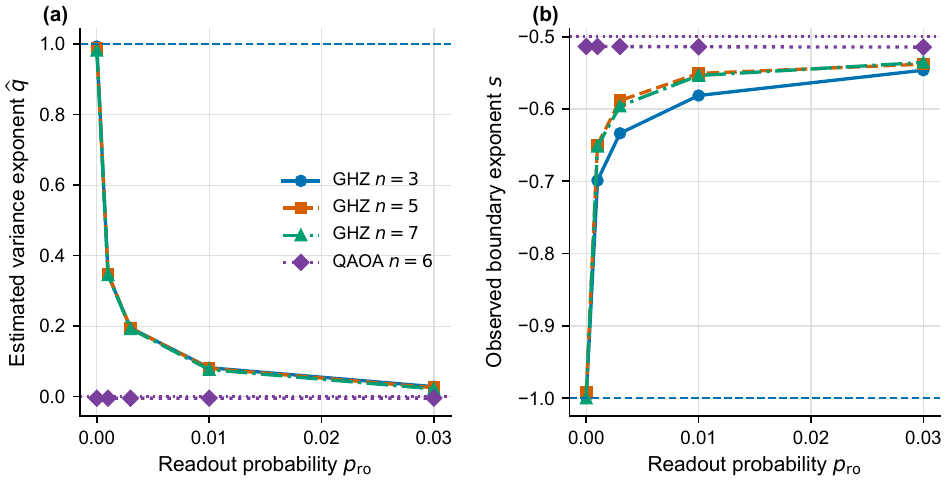}
\caption{Readout/SPAM as an effective measurement-regime change, with known-readout-probability calibration (seed 42). (a) Estimated variance exponent $\widehat q$ versus readout probability $p_{\mathrm{ro}}$; (b) observed boundary exponent $s$. GHZ moves from a deterministic $q\approx1$ regime at $p_{\mathrm{ro}}=0$ toward a variance-floor $q\approx0$ regime as $p_{\mathrm{ro}}$ increases, while QAOA remains in the $q\approx0$ class.}
\label{fig:readout}
\end{figure}

\textbf{Coherent over-rotation and a Monte Carlo check.} Under coherent over-rotation, QAOA $n=8$ gives $s=-0.4855$ ($R^2=0.9987$). For GHZ $n=5$, $\Delta_{\mathrm{MSE}}$ never becomes positive in the scanned window, a degenerate cancellation case without a crossing. Finally, a direct Monte Carlo spot check tested the plug-in procedure itself: at 12 points (GHZ $n=5$ and QAOA $n=8$; $\epsilon\in\{10^{-4},10^{-3},10^{-2}\}$; $B\in\{1000,16000\}$), the distribution of $\Delta_{\mathrm{MSE}}$ obtained from high-shot calibration plus bootstrap was compared with 100 independent Aer repetitions per point, and the two empirical confidence intervals overlapped at all 12 points, with standardized differences of at most $0.31$. The check confirms the shortcut; differences of a few percent, such as the plug-in offset of Section~\ref{sec:diagnostic}, lie below its resolution near the crossing, where $\Delta_{\mathrm{MSE}}$ is close to zero.

These controls confirm that the exponent follows the variance class of the measurement protocol; the deviations come from finite budget windows or from violated assumptions.

\FloatBarrier
\subsection{IBM hardware runs}
\label{sec:ibm}

The IBM Quantum runs test the full pipeline---circuits, folding, expectation-value extraction, ZNE combination, and raw-count storage---on real hardware and compare ZNE with the unmitigated estimator at equal budget. Because the backend noise strength is not a controlled parameter, they do not yield boundary exponents. The folding factors are nominal: folding is applied to the logical circuit before transpilation, so routing SWAP gates inserted afterwards are not folded. For GHZ chains, which need no routing, the transpiled two-qubit gate count grows exactly by the nominal factor $3$. For QAOA it grows on average by only $1.93$--$2.06$ at nominal $\lambda=3$ (minimum $1.62$) and $1.59$--$1.65$ at nominal $\lambda=2$ (minimum $0.93$), depending on the backend, so Richardson extrapolation with the nominal factors under-corrects the QAOA bias. The initial physical layout is identical at both scales in all 101 hardware instances.

The runs comprise GHZ stabilizer tests (G1--G3 and two runs on \texttt{ibm\_fez}), two single QAOA $n=5$ MaxCut instances (Q1, Q2), batches of five and ten random graphs on \texttt{ibm\_marrakesh} (Q3, Q4), two repeats of these protocols on \texttt{ibm\_fez} (R1 with 4000 and R2 with 8000 shots per scale), and a backend and scale-factor grid over ten new QAOA graphs (E1--E6); Appendix~\ref{app:hardware-details} gives the details. In the raw comparisons, which use $N$ shots at each of the two scale factors and so give ZNE twice the shots of the unmitigated estimate, ZNE moved the G1--G3 estimates toward the noiseless target and reduced the QAOA error in both single instances and in 14 of the 15 Q3 and Q4 graphs.

\textbf{Equal-budget comparison.} Table~\ref{tab:ibm-equal-budget} evaluates every hardware run at equal total budget from the recorded counts, without new jobs. The unmitigated estimator receives $B=N$ shots at $\lambda=1$ and first-order Richardson receives $N/2$ shots at each scale. $\Delta_{\mathrm{MSE}}(N)$ is computed with Eq.~\eqref{eq:plugin-mse}, with the calibration term $v_j/N$ removed from the squared biases, and its 95\% interval comes from a basic (pivotal) multinomial bootstrap of the counts, so it reflects shot noise within each job rather than drift between jobs. Equal-budget absolute errors are averaged over random $N/2+N/2$ subsamples of the counts.

\begin{table}[t]
\centering
\footnotesize
\setlength{\tabcolsep}{3pt}
\renewcommand{\arraystretch}{1.1}
\caption{Equal-budget comparison of all IBM hardware runs. $N$ is the number of recorded shots per scale factor. MAE columns are mean absolute errors relative to the noiseless value: unmitigated with $N$ shots; ZNE as reported, with $N$ shots per scale ($2N$ in total); and ZNE at equal budget, with $N/2$ shots per scale. The last two columns refer to the equal-budget plug-in $\Delta_{\mathrm{MSE}}(N)$: the number of instances whose 95\% interval lies above zero (ZNE lowers the MSE) or below zero (ZNE raises it), and the mean over instances with its 95\% confidence interval (CI), in units of $10^{-3}$. The GHZ row combines G1--G3 with the two \texttt{ibm\_fez} runs; R1 and R2 are the repeats on \texttt{ibm\_fez} (GHZ $n=3,5$, ring, five-graph batch, and, for R2, ten-graph batch); E1--E6 are the backend and scale-factor grid of Appendix~\ref{app:hardware-details}.}
\label{tab:ibm-equal-budget}
\begin{tabularx}{\textwidth}{@{}l l c r r r r c X@{}}
\toprule
Runs & Backend & Scales & $N$ & \multicolumn{3}{c}{MAE} & Resolved & Mean $\Delta_{\mathrm{MSE}}(N)$ \\
\cmidrule(lr){5-7}
 & & & & noisy & ZNE (2N) & ZNE (N) & $>0$ / $<0$ & [95\% CI] \\
\midrule
GHZ G1--G3 + 2 & 3 backends & $[1,3]$ & 4000 & 0.1405 & 0.0721 & 0.0726 & 4/5 \,/\, 0/5 & 15.0 [13.0, 17.2] \\
QAOA Q1, Q2 & \texttt{fez}, \texttt{marr.} & $[1,3]$ & 4000 & 0.0504 & 0.0287 & 0.0290 & 1/2 \,/\, 0/2 & 2.0 [1.5, 2.7] \\
Five-graph Q3 & \texttt{marrakesh} & $[1,3]$ & 4000 & 0.0536 & 0.0353 & 0.0359 & 3/5 \,/\, 0/5 & 2.0 [1.7, 2.3] \\
Ten-graph Q4 & \texttt{marrakesh} & $[1,3]$ & 4000 & 0.0645 & 0.0281 & 0.0308 & 9/10 \,/\, 0/10 & 3.4 [3.0, 3.8] \\
Repeat R1 & \texttt{fez} & $[1,3]$ & 4000 & 0.1640 & 0.1147 & 0.1149 & 6/8 \,/\, 1/8 & 12.2 [10.0, 14.4] \\
Repeat R2 & \texttt{fez} & $[1,3]$ & 8000 & 0.1787 & 0.1247 & 0.1247 & 18/18 \,/\, 0/18 & 17.0 [16.5, 17.5] \\
Grid E1 & \texttt{fez} & $[1,3]$ & 16000 & 0.1422 & 0.0920 & 0.0920 & 10/10 \,/\, 0/10 & 12.5 [12.1, 12.9] \\
Grid E2 & \texttt{fez} & $[1,2]$ & 16000 & 0.1578 & 0.0943 & 0.0952 & 9/10 \,/\, 0/10 & 15.0 [14.1, 16.0] \\
Grid E3 & \texttt{kingston} & $[1,3]$ & 16000 & 0.0337 & 0.0248 & 0.0250 & 6/10 \,/\, 0/10 & 0.72 [0.61, 0.84] \\
Grid E4 & \texttt{kingston} & $[1,2]$ & 16000 & 0.0346 & 0.0305 & 0.0310 & 4/10 \,/\, 0/10 & 0.22 [$-0.03$, 0.52] \\
Grid E5 & \texttt{marrakesh} & $[1,2]$ & 16000 & 0.0486 & 0.0361 & 0.0363 & 5/10 \,/\, 0/10 & 1.7 [1.4, 2.1] \\
Grid E6 & \texttt{marrakesh} & $[1,2]$ & 8000 & 0.0113 & 0.0218 & 0.0245 & 0/3 \,/\, 0/3 & $-0.74$ [$-1.43$, 0.63] \\
\bottomrule
\end{tabularx}
\end{table}

At equal budget, ZNE still lowers the MSE in most instances. Of the 101 hardware instances, 75 show a resolved MSE reduction and one a resolved increase: graph seed 202 of the five-graph batch in R1, for which $\Delta_{\mathrm{MSE}}(N)=-0.0107$ $[-0.0165,-0.0042]$; the same graph on the same backend is resolved in favour of ZNE in R2, submitted minutes later. The resolved reductions concentrate where the unmitigated error is large: in the GHZ, R1, R2, E1, and E2 batches, whose mean absolute errors exceed $0.14$, 47 of 51 instances are resolved, against 28 of 50 in the remaining batches, whose mean absolute errors are below $0.07$. This is the bias--variance trade-off behind the boundary: ZNE helps where the squared bias it removes exceeds the variance it adds. Locating the boundary itself on hardware requires a calibrated noise multiplier, for example through folding after routing, and matched budget scans.

\section{Discussion}
\label{sec:discussion}

\textbf{Relation to prior work.} Foundational work established ZNE and Richardson mitigation \cite{Temme2017,LiBenjamin2017,Endo2018} and practical noise scaling \cite{GiurgicaTiron2020,He2020,Pascuzzi2022,RussoMari2024}, and QEM is reviewed in \cite{Cai2023,Endo2021}. Related QEM methods include symmetry expansion, learning-based mitigation, and Clifford-data mitigation \cite{Cai2021,Strikis2021,Czarnik2021}, measurement-error and readout mitigation \cite{Bravyi2021,Nation2021}, and hardware demonstrations on noisy processors \cite{Kandala2019,Kim2023}. Benchmarking and finite-shot studies \cite{Krebsbach2022,Bultrini2023,Russo2023} quantify performance, optimize extrapolation, or estimate sampling overhead, and sampling-cost lower bounds and error-scaling analyses \cite{Takagi2022,Takagi2023,Tsubouchi2023,Qin2023,Quek2024} address worst-case and large-system limits of QEM. Closest to the present work, Mohammadipour and Li analyze ZNE with bias and variance bounds, sample-complexity estimates, and least-squares extrapolation that controls overfitting and measurement noise \cite{MohammadipourLi2025}. The question asked here is complementary: for a fixed noisy-versus-Richardson pair and a fixed budget, where is the perturbative crossing at which ZNE starts to lower the MSE? The answer is an explicit local law whose constants keep the dependence on circuit size and protocol that large-system bounds do not resolve.

\textbf{Practical reading.} The result turns the choice between ZNE and the unmitigated estimator into a comparison of $\epsilon$ with $\epsilon^*(B)$, and placing $\epsilon$ relative to $\epsilon^*(B)$ needs few inputs. Whether the ideal measured variable is deterministic fixes the exponent (Proposition~\ref{prop:variance-exponents}). With scale factors $(1,3)$ and uniform allocation the constants are explicit: ZNE helps above $\epsilon^*(B)\simeq10/(|\alpha|B)$ for deterministic binary observables (Proposition~\ref{prop:depol-contraction-a5}) and above $\epsilon^*(B)\simeq2\sqrt{v(0)/B}/|\alpha|$ for measurements with ideal variance $v(0)>0$, for which $K_{0,1}=4v(0)$. A local estimate of the bias slope $\alpha$, and of $v(0)$ when it is positive, is therefore enough.

\textbf{Open questions.} The theorem concerns the first small-noise crossing. At larger noise, higher-order bias, imperfect noise scaling, coherent accumulation, non-Markovian effects, or nonlinear extrapolation error can create further crossings. The constants $D_p$ and $K_q$ depend on circuit size, so a statement uniform in $n$, which would connect the local law to the large-system sampling-cost bounds above, requires additional structure. Adaptive or correlated sampling schemes beyond the effective-covariance form of Remark~\ref{rem:rounding-correlated} are also outside the present analysis; they are the setting in which the critical and supercritical regimes of Theorem~\ref{thm:master} could arise for real estimators.

\section{Conclusion}
\label{sec:conclusion}

This work identified the finite-shot MSE boundary of fixed Richardson zero-noise extrapolation: the lower local crossing at which ZNE starts to reduce the mean-squared error relative to the unmitigated estimator at a fixed measurement budget. When the leading squared-bias improvement is $D_p\epsilon^{2p}$ and the leading excess-variance penalty is $K_q\epsilon^q/B$, the boundary shrinks as $C_{p,q}B^{-1/(2p-q)}$ for $0\le q<2p$, becomes a budget threshold at $q=2p$, and has no shrinking lower branch for $q>2p$. For fixed Richardson ZNE with linear leading bias the exponent is set by the measurement protocol, which separates deterministic stabilizer-type measurements ($B^{-1}$) from nonzero-variance Hamiltonian-type measurements ($B^{-1/2}$), while Richardson order and allocation change only constants. For finite-outcome shot measurements with smooth noise and independent sampling, these two classes are exhaustive and only the shrinking regime occurs.

Qiskit Aer simulations reproduce this separation for GHZ and QAOA circuits, together with the predicted constants, the synthetic higher-order, critical, and supercritical regimes, and the allocation invariance of the exponent; the deviations seen under amplitude damping, at QAOA depth $r=3$, and with readout error trace back to a finite budget window, a nearly vanishing linear bias, and a broken assumption or change of variance class, respectively. On IBM Quantum hardware, ZNE lowers the MSE at equal budget in 75 of 101 instances, mostly where the unmitigated error is large, and raises it in one. Larger circuit families, calibrated readout protocols, and hardware experiments with controlled noise scaling and matched budget scans are the natural next tests.

\appendix

\section{Proof of Theorem~\ref{thm:master}}
\label{app:master-proof}

\begin{lemma}[Stability of the balanced root]
\label{lem:root-stability}
Let $f_{p,q}(x)=D_px^{2p}-K_qx^q$ with $D_p>0$, $K_q>0$, $p\ge1$, and $0\le q<2p$. Then $C_{p,q}$ of Eq.~\eqref{eq:master-constant} is the only positive root of $f_{p,q}$; it is simple, $f_{p,q}<0$ on $(0,C_{p,q})$, and $f_{p,q}>0$ on $(C_{p,q},\infty)$. If $G_B\to f_{p,q}$ and $G'_B\to f'_{p,q}$ uniformly on a compact neighborhood of $C_{p,q}$, then, for all sufficiently large $B$, $G_B$ has a unique root $x_B$ in a sufficiently small neighborhood of $C_{p,q}$, and $x_B\to C_{p,q}$. If, in addition, $G_B-f_{p,q}=O(B^{-\theta})$ uniformly on that neighborhood, then
\[
x_B=C_{p,q}+O(B^{-\theta}).
\]
\end{lemma}

\begin{proof}
For $x>0$, $f_{p,q}(x)=x^q(D_px^{2p-q}-K_q)$ with $2p-q>0$, which gives the root and the sign pattern. At the root, $K_qC_{p,q}^q=D_pC_{p,q}^{2p}$, so
\begin{equation}
f'_{p,q}(C_{p,q})
=2pD_pC_{p,q}^{2p-1}-qK_qC_{p,q}^{q-1}
=D_p(2p-q)C_{p,q}^{2p-1}>0,
\label{eq:simple-root-derivative}
\end{equation}
and the root is simple. Choose a small interval $U$ around $C_{p,q}$ on which $f'_{p,q}$ is positive and whose endpoints have opposite signs under $f_{p,q}$. Uniform convergence preserves these signs and the monotonicity for large $B$, giving existence and local uniqueness of a root $x_B\in U$. Every convergent subsequence of such roots must converge to the only zero of $f_{p,q}$ in $U$, so $x_B\to C_{p,q}$. For the rate, $G_B(x_B)=0$ gives $f_{p,q}(x_B)=O(B^{-\theta})$, and the expansion $f_{p,q}(x_B)=f'_{p,q}(C_{p,q})(x_B-C_{p,q})+O((x_B-C_{p,q})^2)$ around the simple root gives $x_B-C_{p,q}=O(B^{-\theta})$.
\end{proof}

\begin{proof}[Proof of Theorem~\ref{thm:master}]
For case (i), take $\omega$, $\epsilon_{\mathrm R}$, $L_b$, $L_v$ from A5, write $\epsilon=xB^{-\omega}$, and fix $0<m<C_{p,q}<M$. Since $2p\omega=q\omega+1$,
\begin{equation}
B^{2p\omega}\Delta_{\mathrm{MSE}}(xB^{-\omega},B)=f_{p,q}(x)+\eta_B(x).
\label{eq:rescaled-difference}
\end{equation}
The interval $(0,\bar\epsilon]$ splits into three regions.

\emph{Below the balance scale}, $0<x\le m$. Dividing Eq.~\eqref{eq:leading-mse-balance} by $\epsilon^q/B$,
\[
\frac{\Delta_{\mathrm{MSE}}}{\epsilon^q/B}
=D_px^{2p-q}-K_q+\xi,
\qquad
|\xi|\le L_bx^{2p-q}\epsilon^{\delta_b}+L_v\epsilon^{\delta_v}.
\]
Since $D_pm^{2p-q}<K_q$ and $\epsilon\le mB^{-\omega}\to0$, the right side is negative for all $x\in(0,m]$ once $B$ is large.

\emph{Above the balance scale}, $MB^{-\omega}\le\epsilon\le\bar\epsilon$. Dividing by $\epsilon^{2p}$,
\[
\frac{\Delta_{\mathrm{MSE}}}{\epsilon^{2p}}
=D_p-K_qx^{q-2p}+\xi',
\qquad
|\xi'|\le L_b\epsilon^{\delta_b}+L_vx^{q-2p}\epsilon^{\delta_v}.
\]
Since $x\ge M$ and $D_p-K_qM^{q-2p}>0$, choosing $\bar\epsilon\le\epsilon_{\mathrm R}$ small enough makes the right side positive for every $B$.

\emph{At the balance scale}, $m\le x\le M$. By A5, $\eta_B\to0$ and $\partial_x\eta_B\to0$ uniformly on $[m,M]$, so Lemma~\ref{lem:root-stability} with $G_B=f_{p,q}+\eta_B$ gives, for large $B$, a neighborhood $U\subset(m,M)$ of $C_{p,q}$ in which the rescaled difference is increasing and has exactly one zero $x_B\to C_{p,q}$. On $[m,M]\setminus U$, $|f_{p,q}|$ is bounded below by a positive constant, so the rescaled difference keeps the sign of $f_{p,q}$ there: negative below $U$ and positive above it.

Together, the three regions show that, for large $B$, $\Delta_{\mathrm{MSE}}(\cdot,B)$ has exactly one zero $\epsilon^*(B)=x_BB^{-\omega}$ in $(0,\bar\epsilon]$, with the stated signs. For the rate, the remainder bound in A5 gives
\[
|\eta_B(x)|\le L_bx^{2p+\delta_b}B^{-\omega\delta_b}+L_vx^{q+\delta_v}B^{-\omega\delta_v}=O(B^{-\theta})
\]
uniformly on $[m,M]$, and the rate part of Lemma~\ref{lem:root-stability} gives $x_B=C_{p,q}+O(B^{-\theta})$, which is Eq.~\eqref{eq:master-law}.

For case (ii), $q=2p$ and, for fixed $B$, the remainder in Eq.~\eqref{eq:leading-mse-balance} is $O(\epsilon^{2p+\delta_b})+O(\epsilon^{2p+\delta_v}/B)=o(\epsilon^{2p})$. This gives the displayed expansion, whose leading coefficient changes sign at $B^*=K_q/D_p$.

For case (iii), write
\[
\Delta_{\mathrm{MSE}}(\epsilon,B)
=\epsilon^{2p}\left[D_p-\frac{K_q}{B}\epsilon^{q-2p}+o(1)\right].
\]
Since $q-2p>0$, the term $B^{-1}\epsilon^{q-2p}$ tends to zero as $\epsilon\downarrow0$ for each fixed $B>0$, so the bracket is positive for sufficiently small positive $\epsilon$.
\end{proof}

\section{Finite-budget bracket}
\label{app:aux-boundary}

Theorem~\ref{thm:master}(i) is asymptotic in $B$. The following bracket holds at every budget above an explicit threshold.

\begin{proposition}[Finite-budget bracket]
\label{prop:finite-budget-bracket}
Assume A5 in the subcritical regime $0\le q<2p$, with $\omega$, $\epsilon_{\mathrm R}$, $L_b$, $L_v$ as in A5, $C=C_{p,q}$ as in Eq.~\eqref{eq:master-constant}, and $f_{p,q}$ as in Lemma~\ref{lem:root-stability}, and suppose that $R_p(\epsilon,B)$ is continuous in $\epsilon$. Fix $0<\tau<1$, and define $x_-=(1-\tau)C$, $x_+=(1+\tau)C$, and
\[
\mathfrak m_\tau=
\min\{K_qx_-^q-D_px_-^{2p},\;D_px_+^{2p}-K_qx_+^q\}>0.
\]
If
\begin{align}
B\ge B_0(\tau):=\max\Bigg\{&
\left(\frac{4L_bx_+^{2p+\delta_b}}{\mathfrak m_\tau}\right)^{1/(\omega\delta_b)},\notag\\
&\left(\frac{4L_vx_+^{q+\delta_v}}{\mathfrak m_\tau}\right)^{1/(\omega\delta_v)},
\left(\frac{x_+}{\epsilon_{\mathrm R}}\right)^{1/\omega}
\Bigg\},
\label{eq:finite-budget-threshold}
\end{align}
then
\[
\Delta_{\mathrm{MSE}}(x_-B^{-\omega},B)<0<\Delta_{\mathrm{MSE}}(x_+B^{-\omega},B),
\]
so $\Delta_{\mathrm{MSE}}(\cdot,B)$ has a sign-changing zero in the bracket
\begin{equation}
(1-\tau)CB^{-\omega}\le \epsilon
\le (1+\tau)CB^{-\omega}.
\label{eq:finite-budget-bracket}
\end{equation}
If, in addition, $\tau$ is chosen so that $f'_{p,q}>0$ on $[x_-,x_+]$ and the rescaled remainder $\eta_B$ of A5 satisfies
\[
|\partial_x\eta_B(x)|
\le \frac12\min_{y\in[x_-,x_+]} f'_{p,q}(y),
\qquad x\in[x_-,x_+],
\]
the zero in this bracket is unique.
\end{proposition}

\begin{proof}
By Eq.~\eqref{eq:rescaled-difference} and the remainder bound in A5,
\[
|\eta_B(x)|\le L_bx^{2p+\delta_b}B^{-\omega\delta_b}
+L_vx^{q+\delta_v}B^{-\omega\delta_v}
\]
for $xB^{-\omega}\le\epsilon_{\mathrm R}$. By the definition of $B_0(\tau)$, $|\eta_B(x_\pm)|\le \mathfrak m_\tau/2$. Since
$f_{p,q}(x_-)\le-\mathfrak m_\tau$ and $f_{p,q}(x_+)\ge\mathfrak m_\tau$, the signs persist at the two endpoints, and the intermediate value theorem gives a zero inside the bracket. The uniqueness statement follows because the stated $\partial_x\eta_B$ bound makes $f'_{p,q}+\partial_x\eta_B$ positive on $[x_-,x_+]$, so the rescaled MSE difference is strictly increasing there.
\end{proof}

\section{IBM hardware run details}
\label{app:hardware-details}

This appendix gives the hardware details summarized in Section~\ref{sec:ibm}. Backend names, shot counts, graph seeds, and IBM job metadata are in the reproducibility supplement.

Tables~\ref{tab:ibm-ghz-details} and~\ref{tab:ibm-qaoa-details} report the raw comparisons of Section~\ref{sec:ibm}; Table~\ref{tab:ibm-equal-budget} gives the equal-budget comparison with bootstrap intervals for every run. For example, the \texttt{ibm\_marrakesh} single-instance QAOA run (Q2) and GHZ $n=3$ on \texttt{ibm\_kingston} (G1) improve in the raw comparison but are unresolved at equal budget.

\begin{table}[htbp]
\centering
\footnotesize
\setlength{\tabcolsep}{4pt}
\renewcommand{\arraystretch}{1.08}
\caption{IBM GHZ hardware runs (raw comparison): expectation values at the two scale factors and the ZNE estimate.}
\label{tab:ibm-ghz-details}
\begin{tabular*}{\textwidth}{@{\extracolsep{\fill}}ll l r r r l@{}}
\toprule
Run & Case & Backend & $\lambda=1$ & $\lambda=3$ & ZNE & Interpretation \\
\midrule
G1 & GHZ $n=3$ & \texttt{ibm\_kingston} & 0.9230 & 0.9165 & 0.92625 & improved \\
G2 & GHZ $n=5$ & \texttt{ibm\_kingston} & 0.8640 & 0.8230 & 0.88450 & improved \\
G3 & GHZ $n=5$ & \texttt{ibm\_marrakesh} & 0.8430 & 0.7640 & 0.88250 & improved \\
\bottomrule
\end{tabular*}
\end{table}

Two additional GHZ runs on \texttt{ibm\_fez} gave extrapolated values $0.96225$ for GHZ $n=3$ and $1.01575$ for GHZ $n=5$. The latter slightly exceeds the physical Pauli range and is reported without clipping because linear Richardson extrapolation is not range-preserving.

\begin{table}[htbp]
\centering
\footnotesize
\setlength{\tabcolsep}{4pt}
\renewcommand{\arraystretch}{1.08}
\caption{IBM QAOA hardware runs (raw comparison). Single-instance entries are absolute errors relative to the ideal simulator value; batch rows report averages over fixed graph seeds (Q3 and Q4: five and ten graphs, 4000 shots per scale factor).}
\label{tab:ibm-qaoa-details}
\begin{tabularx}{\textwidth}{@{}lL{0.18\textwidth}l r L{0.13\textwidth}r X@{}}
\toprule
Run & Case & Backend & $\lambda=1$ err. & $\lambda=3$ err./status & ZNE err. & Interpretation \\
\midrule
Q1 & Ring $n=5$, $r=1$ & \texttt{ibm\_fez} & 0.0726 & 0.1538 & 0.0320 & improved \\
Q2 & $n=5$, $r=1$ & \texttt{ibm\_marrakesh} & 0.0282 & 0.0336 & 0.0255 & improved \\
Q3 & Five-graph batch & \texttt{ibm\_marrakesh} & 0.05357 & degraded 5/5 & 0.03531 & improved 5/5 \\
Q4 & Ten-graph batch & \texttt{ibm\_marrakesh} & 0.06445 & degraded 10/10 & 0.02810 & improved 9/10 \\
\bottomrule
\end{tabularx}
\end{table}

The two repeats on \texttt{ibm\_fez} reran the GHZ $n=3,5$, ring, and five-graph protocols with 4000 shots per scale factor (R1) and, adding the ten-graph batch, with 8000 shots per scale factor (R2), using the same circuits and graph seeds as above. The unmitigated errors on \texttt{ibm\_fez} were larger than on \texttt{ibm\_marrakesh} (Table~\ref{tab:ibm-equal-budget}). Job identifiers for R1, R2, and E1--E6 are listed in the supplement.

The grid E1--E6 tests whether this behavior depends on the backend, the scale choice, or the graph set: it uses ten new graph seeds and compares Richardson scale factors $[1,3]$ and $[1,2]$ across IBM backends. The first-positive mean $\Delta_{\mathrm{MSE}}$ column in Table~\ref{tab:ibm-qaoa-extra-details} is a finite-budget MSE diagnostic over saved shot subsamples.

\begin{table}[htbp]
\centering
\footnotesize
\setlength{\tabcolsep}{3pt}
\renewcommand{\arraystretch}{1.08}
\caption{IBM QAOA backend and scale-factor grid. All full-grid rows use QAOA MaxCut $n=5$, $r=1$, 16000 shots per folded circuit, and the same ten new graph seeds $[1707,1808,1909,2001,2102,2203,2304,2405,2506,2607]$. The small-baseline row used the first three of these seeds with 8000 shots. "Improved" and the MAE column are raw comparisons, in which ZNE uses twice the shots of the unmitigated estimate; the last column summarizes the equal-budget comparison of Table~\ref{tab:ibm-equal-budget}.}
\label{tab:ibm-qaoa-extra-details}
\begin{tabularx}{\textwidth}{@{}l L{0.15\textwidth} l r L{0.15\textwidth}L{0.17\textwidth} X@{}}
\toprule
Run & Backend & Scales & Improved & MAE $\lambda=1\to$ ZNE & First positive mean $\Delta_{\mathrm{MSE}}$ & Interpretation \\
\midrule
E1 & \texttt{ibm\_fez} & $[1,3]$ & $10/10$ & $0.14217\to0.09198$ & $B=200$ & 10/10 resolved at equal budget \\
E2 & \texttt{ibm\_fez} & $[1,2]$ & $10/10$ & $0.15783\to0.09432$ & $B=500$ & 9/10 resolved at equal budget \\
E3 & \texttt{ibm\_kingston} & $[1,3]$ & $7/10$ & $0.03368\to0.02476$ & $B=2000$ & 6/10 resolved; none resolved negative \\
E4 & \texttt{ibm\_kingston} & $[1,2]$ & $5/10$ & $0.03462\to0.03046$ & $B=16000$ & 4/10 resolved; batch mean unresolved \\
E5 & \texttt{ibm\_marrakesh} & $[1,2]$ & $6/10$ & $0.04858\to0.03609$ & $B=2000$ & 5/10 resolved; batch mean positive \\
E6 & \texttt{ibm\_marrakesh} & $[1,2]$ & $0/3$ & $0.01134\to0.02178$ & not reached & small baseline error; unresolved \\
\bottomrule
\end{tabularx}
\end{table}

\section*{Acknowledgements}
The author gratefully acknowledges the personal support and encouragement of his family. The author also acknowledges the open-source software ecosystems of Qiskit and Mitiq, and the IBM Quantum cloud hardware services used for the hardware experiments.

\section*{Declarations}
\textbf{Funding.} No funding was received for this work.

\noindent\textbf{Competing interests.} The author declares no financial or non-financial competing interests.

\noindent\textbf{Ethics approval.} Not applicable. This study does not involve human participants, human data, animals, or clinical material.

\noindent\textbf{Data availability.} Raw counts for all simulations and IBM hardware runs, processed summaries, bootstrap outputs, crossing tables, variance-exponent fits, readout-regime experiments, IBM job metadata, figure inputs, and the analysis scripts of Sections~\ref{sec:diagnostic} and~\ref{sec:empirical} are available in Zenodo as version 1.0.3 of the reproducibility supplement \cite{ScavinoZenodo2026}, DOI: \href{https://doi.org/10.5281/zenodo.23405717}{10.5281/zenodo.23405717}; all versions are collected under the concept DOI \href{https://doi.org/10.5281/zenodo.20057160}{10.5281/zenodo.20057160}.

\noindent\textbf{Code availability.} The same Zenodo record archives the Qiskit/Mitiq simulation code, bootstrap routines, hardware scripts, dependency information, and figure-generation code. The corresponding tagged GitHub source is linked from the Zenodo record.

\noindent\textbf{ORCID.} Vicenzo Scavino Alfaro: 0009-0000-2472-9785.

\noindent\textbf{Author contribution.} Vicenzo Scavino Alfaro: Conceptualization, Methodology, Software, Validation, Formal analysis, Investigation, Data curation, Writing--original draft, Writing--review and editing, and Visualization.

\noindent\textbf{Use of AI-assisted tools.} AI-assisted tools were used for language editing, LaTeX formatting, code-editing support, and consistency review. The author verified the final text, code, calculations, citations, figures, and conclusions, and takes full responsibility for the work.

\end{document}